\documentclass[runningheads]{llncs}

\usepackage{adjustbox}
\usepackage{algorithmic}
\usepackage[ruled,linesnumbered]{algorithm2e}

\usepackage{amsmath,amsfonts,amssymb}

\usepackage{booktabs}
\usepackage{bm}
\usepackage{diagbox}
\usepackage{float}
\usepackage{graphicx}
\usepackage{hyperref}
\hypersetup{colorlinks=true, linkcolor=blue, citecolor=blue, urlcolor=blue}
\usepackage{listings}
\usepackage{mathrsfs}
\usepackage{mathtools}
\usepackage{multirow}
\usepackage[sort&compress,numbers]{natbib}
\usepackage{subcaption}
\usepackage{thmtools}
\usepackage{thm-restate}
\usepackage{tikz}
\usetikzlibrary{arrows.meta,positioning,calc,fit,backgrounds,decorations.pathreplacing,shapes.geometric}
\usepackage{xcolor}
\usepackage[margin=1in]{geometry}

\usepackage[T1]{fontenc}

\spnewtheorem{assumption}{Assumption}{\bfseries}{\itshape}
\spnewtheorem{fact}{Fact}{\bfseries}{\itshape}

\usepackage[capitalize,noabbrev]{cleveref}
\crefname{assumption}{Assumption}{Assumptions}
\Crefname{assumption}{Assumption}{Assumptions}
\crefname{fact}{Fact}{Facts}
\Crefname{fact}{Fact}{Facts}
\crefname{example}{Example}{Examples}
\Crefname{example}{Example}{Examples}

\definecolor{pku-red}{RGB}{139,0,18}

\definecolor{croot}{RGB}{31,119,180}   
\definecolor{cleaf}{RGB}{46,160,67}     
\definecolor{cref}{RGB}{240,110,110}    
\definecolor{cplayer}{RGB}{26,188,188}  
\definecolor{creall}{RGB}{230,145,25}   
\definecolor{cdyn}{RGB}{31,119,180}     
\definecolor{cdb}{RGB}{245,224,140}     

\newcommand{\bbE}{\mathbb{E}}

\newcommand{\bbR}{\mathbb{R}}

\newcommand{\bbU}{\mathbb{U}}

\newcommand{\calF}{\mathcal{F}}

\newcommand{\calH}{\mathcal{H}}
\newcommand{\calI}{\mathcal{I}}

\newcommand{\calL}{\mathcal{L}}

\newcommand{\calP}{\mathcal{P}}

\newcommand{\calT}{\mathcal{T}}

\newcommand{\calV}{\mathcal{V}}

\newcommand{\bmp}{{\bm{p}}}

\newcommand{\bms}{{\bm{s}}}

\newcommand{\bmu}{{\bm{u}}}
\newcommand{\bmv}{{\bm{v}}}
\newcommand{\bmw}{{\bm{w}}}
\newcommand{\bmx}{{\bm{x}}}

\newcommand{\one}{\bm{1}}

\DeclareMathOperator*{\argmax}{argmax}

\newcommand{\poly}{\mathrm{poly}}

\newcommand{\ie}{\emph{i.e.}}

\newcommand{\SW}{\mathrm{SW}}

\newcommand{\tbmp}{\tilde{\bmp}}

\newcommand{\tbmx}{\tilde{\bmx}}
\newcommand{\tbms}{\tilde{\bms}}

\newcommand{\tp}{\tilde{p}}

\newcommand{\tx}{\tilde{x}}

\renewcommand{\ts}{\tilde{s}}

\newcommand{\tmu}{{\tilde{\mu}}}
\newcommand{\bmmu}{{\bm{\mu}}}
\newcommand{\tbmmu}{{\tilde{\bmmu}}}

\DeclareMathOperator{\Children}{Children}
\DeclareMathOperator{\Parent}{Parent}
\DeclareMathOperator{\Path}{Path}
\DeclareMathOperator{\Subtree}{Subtree}
\DeclareMathOperator{\Desc}{Desc}
\newcommand{\Model}{\textsc{Model}}

\usepackage[normalem]{ulem}

\crefname{appendix}{Appendix}{Appendices}
\Crefname{appendix}{Appendix}{Appendices}

\usepackage{enumerate}
\usepackage[inline]{enumitem}

\title{Profit Reallocation Mechanisms in Tree-based Data Trading}

\authorrunning{Y. Ma et al.}
\author{Yunxuan Ma\inst{1,*} \and Wu Xin\inst{2,*} \and Jichen Li\inst{3} \and Yusen Zheng\inst{1} \and Xiaotie Deng\inst{\dagger}}
\institute{CFCS, School of Computer Science, Peking University, Beijing, China\\ \email{yunxuanma@pku.edu.cn, yusen@stu.pku.edu.cn} \and School of Economics, Peking University, Beijing, China\\ \email{xinwu@pku.edu.cn} \and Department of Computer Science and Technology, Tsinghua University, Beijing, China\\ \email{jichenli@mail.tsinghua.edu.cn}
}

\begin{document}

\maketitle
\begingroup
\renewcommand{\thefootnote}{}%
\footnotetext{* Equal contribution.}%
\footnotetext{\textdagger~Corresponding author.}%
\endgroup
\begin{abstract}
Markets for data promise to unlock its economic value, yet in practice they remain far less active than expected---one reason is that those who supply data are not rewarded for the value it creates downstream. A defining feature of data is its \emph{replicability}: a buyer can refine purchased data into a new product and resell it to \emph{many} downstream buyers, so a single source seeds a branching cascade of resales that naturally forms a \emph{tree}. Because an upstream seller captures none of this downstream value, its incentive to trade is weakened. Existing works propose \emph{profit reallocation}---returning part of downstream revenue to upstream contributors---as a natural remedy. But whether profit reallocation works on the tree-structured markets that replicable data actually induces has remained open.

To bridge this gap, we develop a principled framework for profit reallocation on tree-structured data markets. We introduce a sequential trading game on a tree and a general class of budget-feasible profit reallocation mechanisms (PRMs) over it. We derive efficient algorithms to compute the induced equilibria---a polynomial-time exact algorithm for discrete valuations and a fully polynomial-time approximation scheme (FPTAS) for continuous ones---via a subtree decomposition technique that tames the potential coupling across a seller's children. We then prove that, under mild assumptions, \emph{any} budget-feasible PRM weakly expands the trades that occur in equilibrium, and any budget-balanced PRM additionally weakly improves social welfare, relative to the baseline that reallocates nothing. Experiments on synthetic markets confirm that these benefits are substantial, persist even when the assumptions fail, and grow with the depth and branching of the tree.

\end{abstract}

\newpage
\section{Introduction}
\label{sec:intro}

Data has become a central production factor of the modern economy, and trading is a primary means by which its value is unlocked and redistributed~\citep{analytics2016age,richter2019data}.
A growing number of platforms---such as Snowflake Data Marketplace~\citep{SnowflakeDocs_MarketplaceAbout}, AWS Data Exchange~\citep{AWSDocs_DataExchangeLanding}, and Dawex~\citep{Dawex_Web_DataMarketplaceModel}---now intermediate the exchange of data products.
Yet, in practice, data markets remain far less active than their potential would suggest~\citep{Koutroumpis2017UnfulfilledPotential,OECD2021ValueOfData,EuropeanParliament2023DataActPressRelease}.
Among the many reasons that make the data market struggle, a widely cited cause is the absence of incentive structures that reward all parties whose data ultimately creates value~\citep{OECD_EASD_2019,azcoitia2022survey,bauer2024designing}.

A substantial body of work studies incentive designs in data markets through auctions~\citep{ghosh2011selling,zhang2020selling,agarwal2024towards} and pricing~\citep{chen2019towards,mi2025multi}.
However, most works treat data as physical goods, overlooking the \emph{replicable} nature of data:
once acquired, a buyer can process the data into a new product and resell it to multiple downstream buyers, so a single source can give rise to a cascade of resales~\citep{JonesTonetti2020NonrivalryData,shiller2013digital,biswas2021incentive}. This can form a tree structure of data flow.
Because an upstream seller captures no direct revenue and benefits only indirectly from these downstream resales, its incentive to trade in the first place is weakened.
Recent work proposes \emph{profit reallocation}---redistributing part of downstream revenue back to upstream contributors---as a remedy for this incentive problem~\citep{xin2025tbds,ma2026incentivizing}, yet these analyses assume that each seller resells to at most one buyer.
The design and effects of profit reallocation on the tree-structured markets that replicable data naturally induces remain unexplored.

To close this gap, we develop the first principled framework for profit reallocation on \emph{tree-structured} data markets.
We model data trading as a sequential game on a rooted tree in which the root player originates the data and every player may process and resell its data product to multiple children (\cref{sec:model}).
On top of this game we formalize a \emph{profit reallocation mechanism} (PRM): whenever a trade occurs, part of its payment is redistributed to the buyer's ancestors by a budget-feasible redistribution rule---the total amount redistributed never exceeds the payment---so every player benefits from the trade.
Intuitively, a PRM lets the original seller invest in intermediaries using data, consequently earning direct revenue from every downstream transaction its data feeds.

Given a PRM, the first challenge is to understand how strategic players will behave, \ie, to compute the induced equilibrium; throughout we adopt the notion of \emph{perfect Bayesian equilibrium}, which demands optimality at every information set.
Branching makes this computation formidable: a seller prices data product to \emph{every} child, so its pricing action is a vector whose dimension equals its number of children, and its utility couples the entire subtree beneath it---a naive evaluation is exponential in the tree size.
Fortunately, a key structural insight dissolves this difficulty: under the Markovian valuation process, sibling subtrees are conditionally independent given the seller's valuation, so the seller's pricing problem \emph{separates} across its children.
This subtree decomposition collapses equilibrium strategies to single-variable functions and lets a seller optimize each child's price separately, eliminating any joint optimization over the price vector.
Building on it, we design a polynomial-time exact algorithm for discrete valuations and a fully polynomial-time approximation scheme (FPTAS) for continuous valuations (\cref{sec:algo}), turning the evaluation of any PRM into a tractable computation.

Beyond computation, we examine whether data trading actually benefits from profit reallocation, and answer in the affirmative.
We prove that, under mild assumptions, any budget-feasible PRM weakly expands the set of valuation profiles under which each trade occurs, relative to the baseline mechanism that does not reallocate profits; and any budget-balanced PRM, for which payments are fully redistributed rather than partly withheld by the mechanism, additionally weakly improves social welfare (\cref{sec:theoretic}).
Here, too, branching is the crux of the difficulty: rewarding a seller more from \emph{one} child's subtree could in principle distort its prices to the \emph{other} children, breaking trades elsewhere in the tree. We rule this out by carefully interpolating between PRMs so that adjacent PRMs differ on only one subtree at a time.
Finally, our experiments (\cref{sec:exp}) confirm that the advantage of profit reallocation is substantial, persists even when the theoretical assumptions are violated, and grows with the depth and branching of the tree.

In summary, our contributions are:
\begin{itemize}[left=0em]
\item A tree-based data trading game that captures one-to-many resale, together with a general, budget-feasible profit reallocation mechanism on it (\cref{sec:model});
\item Efficient equilibrium computation---a polynomial-time exact algorithm and an FPTAS---that makes any PRM evaluable (\cref{sec:algo});
\item A theoretical guarantee that any budget-feasible PRM weakly expands trade, and any budget-balanced PRM additionally weakly improves welfare, over the baseline mechanism on trees (\cref{sec:theoretic});
\item Experiments confirming that the advantage is robust to violations of the theoretical assumptions and strengthens as the tree size grows (\cref{sec:exp}).
\end{itemize}
Together, these results suggest that deploying profit reallocation can meaningfully invigorate real-world data markets.

\section{Related Work}
\label{sec:related}

\paragraph{Incentives and pricing in data markets.}
Designing incentives for data exchange is a long-standing concern~\citep{liu2024data,azcoitia2022survey}. One line of work sells data or statistics through auctions, often under privacy constraints~\citep{ghosh2011selling,nissim2014redrawing,zhang2020selling,agarwal2024towards}, while another studies pricing of data and machine-learning products~\citep{chen2017model,chen2019towards,mi2025multi,qin2023research}. Incentive mechanisms have also been developed for federated learning and collaborative analytics~\citep{yu2020fairness,li2023martfl,han2024dynamic,zhang2025incentive}, and for revenue or profit sharing among contributors~\citep{cao2017game,bauer2024designing}. These works largely treat data as a good that is consumed by its buyer, and do not model the strategic consequences of buyers becoming sellers themselves.

\paragraph{Replicability and resale of data.}
Unlike rival goods, data is non-rival and can be replicated and resold~\citep{JonesTonetti2020NonrivalryData,shiller2013digital,ali2024reselling}. A handful of works acknowledge that resale can generate value along a transaction chain~\citep{biswas2021incentive,xin2025tbds,ma2026incentivizing}.
Biswas et al.~\cite{biswas2021incentive} design a marketplace that permits resale but do not analyze incentive design for it.
Xin et al.~\cite{xin2025tbds} propose a blockchain-based mechanism, \textsc{TBDS}, that traces transactions from resale and pays upstream contributors a fixed fraction of downstream revenue.
Ma et al.~\cite{ma2026incentivizing} study the incentive effects of profit reallocation where each data owner resells to at most one buyer.

\paragraph{Mechanisms on networks and sequential trade.}
Our reallocation of value to upstream participants is reminiscent of diffusion auctions, referral mechanisms, and multi-level marketing, which reward intermediaries for propagating an item or information through a network~\citep{emek2011mechanisms,li2017mechanism,li2022diffusion,zhang2020redistribution}. A separate line studies equilibrium behavior in sequential bilateral or intermediated trade~\citep{condorelli2017bilateral,manea2018intermediation,sandholm2006sequences}. Our setting differs in three respects: the traded object is replicable data, so a seller does not relinquish it; the focus is on how an external reallocation mechanism reshapes equilibrium trading behavior; and trades occur on a platform among mutually anonymous players via posted prices, rather than on a known social network.

\section{Model}
\label{sec:model}

\subsection{Tree-Based Data Trading Game}
\label{subsec:game}

We consider a data trading game on a rooted tree $\calT$.
The node set $\calI$ is the set of all players, with each directed edge pointing from a parent to a child. Let $r \in \calI$ denote the root, and $\calL \subseteq \calI$ denote the set of leaf nodes. We write $N \coloneqq |\calI|$ for the total number of players. We define the following functions for players:
\begin{itemize}[left=0em]
\item $\Parent(i)$, for $i \in \calI \setminus \{r\}$: the unique parent of $i$ in $\calT$.
\item $\Children(i)$, for $i \in \calI$: the set of children of $i$. For leaf nodes $i \in \calL$, $\Children(i) = \varnothing$.
\item $\Path(i)$, for $i \in \calI$: the set of strict ancestors of $i$, \ie, the nodes on the path from $r$ to $i$ excluding $i$. We have $\Path(r) = \varnothing$.
\item $\Subtree(i)$, for $i \in \calI$: all nodes in the subtree rooted at $i$, including $i$ itself.
\item $\Desc(i) \coloneqq \Subtree(i) \setminus \{i\}$, for $i \in \calI$: the set of strict descendants of $i$.
\end{itemize}
\cref{fig:tree-notation} illustrates these notations on a small instance.

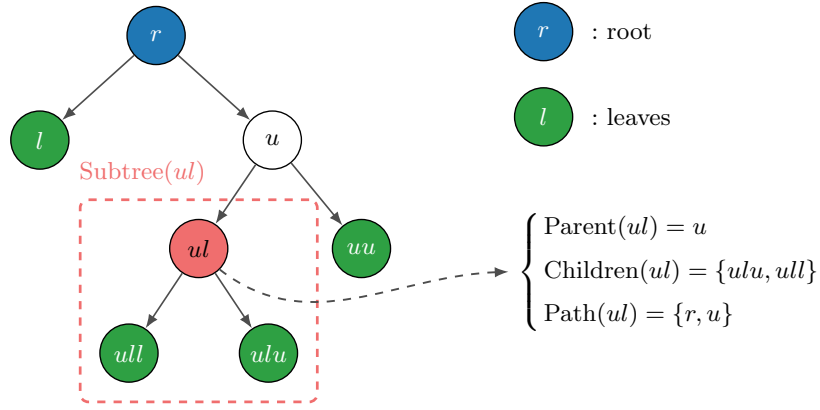
\begin{figure}[t]
\centering
\resizebox{0.66\textwidth}{!}{%
\begin{tikzpicture}[
  >={Latex[length=1.8mm]},
  vtx/.style={circle, draw, line width=0.5pt, minimum size=7.5mm, inner sep=0pt, font=\small},
  rt/.style={vtx, fill=croot, text=white},
  lf/.style={vtx, fill=cleaf, text=white},
  rf/.style={vtx, fill=cref},
  it/.style={vtx, fill=white},
  ed/.style={->, line width=0.6pt, black!70},
]
\node[rt] (r)   at (0,0)        {$r$};
\node[lf] (l)   at (-1.5,-1.35) {$l$};
\node[it] (u)   at (1.5,-1.35)  {$u$};
\node[rf] (ul)  at (0.55,-2.75) {$ul$};
\node[lf] (uu)  at (2.65,-2.75) {$uu$};
\node[lf] (ull) at (-0.35,-4.1) {$ull$};
\node[lf] (ulu) at (1.45,-4.1)  {$ulu$};
\draw[ed] (r)--(l);  \draw[ed] (r)--(u);
\draw[ed] (u)--(ul); \draw[ed] (u)--(uu);
\draw[ed] (ul)--(ull); \draw[ed] (ul)--(ulu);
\begin{scope}[on background layer]
\node[draw=cref, dashed, line width=1pt, rounded corners=3pt,
      fit=(ul)(ull)(ulu), inner sep=7pt] (box) {};
\end{scope}
\node[text=cref, font=\small, anchor=south west] at ([xshift=-3pt,yshift=1.5pt]box.north west) {$\Subtree(ul)$};
\node[rt] (lr) at (5.0,0.05) {$r$};
\node[right=1.2mm of lr, font=\small] {: root};
\node[lf] (ll) at (5.0,-1.05) {$l$};
\node[right=1.2mm of ll, font=\small] {: leaves};
\node[anchor=west, font=\small] (info) at (4.55,-3.05)
  {$\left\{\begin{array}{l}\Parent(ul)=u\\[5pt]\Children(ul)=\{ulu,ull\}\\[5pt]\Path(ul)=\{r,u\}\end{array}\right.$};
\draw[->, dashed, line width=0.7pt, black!70] (ul.south east) to[out=-38,in=182] (info.west);
\end{tikzpicture}%
}
\caption{An illustration of the tree notation on a $7$-player instance. The root $r$ is shown in blue and the leaves $\calL = \{uu, l, ulu, ull\}$ in green. Taking node $ul$ as the reference: its parent is $\Parent(ul) = u$, its children are $\Children(ul) = \{ulu, ull\}$, and its strict ancestors are $\Path(ul) = \{r, u\}$. The red dashed box marks $\Subtree(ul) = \{ul, ulu, ull\}$, so that $\Desc(ul) = \{ulu, ull\}$. Best viewed in color.}
\label{fig:tree-notation}
\end{figure}

Replicable data flow downward along the edges of $\calT$, originating at the root $r$. Once a non-root player $i$ receives data from its parent, it refines the data into a sellable product that it may, in turn, offer to each of its children. Player $i$'s \emph{valuation} is a single scalar $v_i \in \calV_i \subseteq \bbR$: the net worth of the data product to $i$, namely the benefit it derives minus the processing cost, which may be negative. This quantity is observed only by player $i$.
We write $\bmv =\{v_i\}_{i \in \calI}$ as valuation profile, $\calV \coloneqq \prod_{i \in \calI} \calV_i$ for the space of valuation profiles, $\bmv_{-i}$ for the valuation profile other than player $i$'s, and $\calV_{-i} \coloneqq \prod_{k \neq i} \calV_k$ for its domain. We assume each $\calV_i$ is a closed bounded interval.
We let the valuation profile evolve as a Markov process down the tree.

\begin{assumption}[Markovian on Trees]
\label{asp:markov}
The root valuation is drawn as $v_r \sim f_r(\cdot)$. For every non-root player $i \in \calI \setminus \{r\}$, the valuation $v_i$ depends only on $v_{\Parent(i)}$:
\[
v_i \sim f_i(\cdot \mid v_{\Parent(i)}).
\]
The valuation profile $\bmv = \{v_i\}_{i \in \calI}$ has the joint distribution $f(\bmv) = f_r(v_r) \prod_{i\ne r} f_i(v_i|v_{\Parent(i)})$.
We denote the joint distribution of $\bmv$ by $\calF$, and the conditional CDF by $F_i$. Both $\calF$, $F_i$s and $f_i$s are common knowledge.
\end{assumption}

This assumption fits the tree structure: a player $i$'s product's worth is governed by the immediate input (player $\Parent(i)$'s product) from which it is refined, rather than by the raw material that forms the immediate input.
In other words, this assumption captures that the value of a refined product may depend on the value of its raw input.
It also subsumes the classical independent-private-values setting~\citep{myerson1981optimal} as a special case.

\subsection{Player Action}
\label{subsec:action}

Each non-root player $j$ corresponds to a potential trade $j$. The seller $i = \Parent(j)$ first commits to a posted price $p_j \in \bbR_+$ for player $j$; seeing this price, player $j$ responds with a binary decision $x_j \in \{0,1\}$, describing her willingness to buy ($x_j = 1$) or not ($x_j = 0$).

Trade $i$ actually occurs if and only if trade $\Parent(i)$ occurs and player $i$ is willing to buy the data from player $\Parent(i)$. We define the indicator of whether trade $i$ occurs:
\begin{equation}
\label{eq:yi}
y_i \;\coloneqq\; \prod_{j \in \Path(i) \cup \{i\}} x_j, \forall i \in \calI \backslash\{r\},
\end{equation}
where we set $x_r \coloneqq 1$ and $p_r \coloneqq 0$ by convention (the root player initially possesses the data).
If $y_i = 1$, trade $i$ occurs, player $i$ obtains the data, realizes valuation $v_i$, and pays $p_i$; otherwise $y_i = 0$, player $i$ gains nothing and pays nothing. For $m \in \Desc(j)$, we define the \emph{conditional trade indicator}:
\begin{equation}
\label{eq:y-conditional}
y_{m \mid j} \;\coloneqq\; \prod_{\substack{\ell \in \Path(m) \cup \{m\} \\ \ell \notin \Path(j) \cup \{j\}}} x_\ell,
\end{equation}
which is the product of trade decisions along the path from $j$ to $m$, excluding $j$. This captures whether trade $m$ occurs conditional on trade $j$ occurs.

\subsection{Player Strategy}
\label{subsec:strategy}

At each stage, a player observes the public trading history along their path from the root. For player $i$, the observed history is:
\begin{equation}
\label{eq:hi}
h_i = \big\{(j, p_j, x_j)\big\}_{j \in \Path(i) \setminus \{r\}} \cup \{(i, p_i)\},
\end{equation}
which includes all past trade outcomes on the path from $r$ to $i$ and the price $p_i$ currently offered to $i$. We set $h_r = \varnothing$. Let $\calH_i$ denote the set of all possible histories for player $i$.

A strategy for player $i$ consists of:
\begin{itemize}[left=0em]
\item \textbf{Pricing strategy} (for $i \notin \calL$): $\tp_i: \calH_i \times \Children(i) \times \calV_i \to \bbR_+$, where $\tp_i(h_i, j, v_i)$ is the price player $i$ offers to child $j$.
\item \textbf{Buying strategy} (for $i \neq r$): $\tx_i: \calH_i \times \calV_i \to \{0, 1\}$, where $\tx_i(h_i, v_i)$ is the player $i$'s willingness to buy given public trading history and her private valuation $v_i$.
\end{itemize}
We call $(\tbmp, \tbmx) = \big(\{\tp_i\}_{i \notin \calL},\; \{\tx_i\}_{i \neq r}\big)$ a strategy profile.

\subsection{Game Dynamic}
\label{subsec:dynamic}

\begin{definition}[Game Dynamics on Trees]
\label{def:game-evolve}
Given a valuation profile $\bmv$ and a strategy profile $(\tbmp, \tbmx)$, the game evolves recursively from root to leaves:\\
1. Initialize $h_r = \varnothing$.\\
2. For each non-leaf player $i$ and each child $j \in \Children(i)$, player $i$ posts price $p_j = \tp_i(h_i, j, v_i)$ to child $j$.\\
3. Each non-root player $i$ decides $x_i = \tx_i(h_i, v_i)$.\\
4. The game concludes with the terminal history $h = \{(i, p_i, x_i)\}_{i \in \calI \setminus \{r\}}$.
\end{definition}

We denote this mapping from valuation profile $\bmv$, strategy profile $(\tbmp, \tbmx)$ to the terminal history $h$ as $h(\tbmp, \tbmx;\bmv)$.

\subsection{Profit Reallocation Mechanism}
\label{subsec:mechanism}

A \emph{profit reallocation mechanism} (PRM) specifies how the payment from each trade is distributed among upstream players.

\begin{definition}[Profit Reallocation Mechanism on Trees]
\label{def:prm}
A PRM is a function $\pi: (\calI \setminus \{r\}) \times (\calI \setminus \calL) \to \bbR_+$, where $\pi(i,j)$ is the proportion of the payment from trade $i$ allocated to player $j$. That is, when trade $i$ occurs at price $p_i$, player $j$ receives the monetary transfer of $\pi(i,j) \cdot p_i$. The function $\pi$ must satisfy following constraints to make it economically meaningful:
\begin{enumerate}[left=0em]
\item \textbf{Path Only:} $\pi(i,j) > 0$ only if $j \in \Path(i)$. Only ancestors of buyer $i$ are eligible for profit reallocation.
\item \textbf{Budget Feasibility:} $\sum_{j \in \Path(i)} \pi(i,j) \le 1$ for all $i \in \calI \setminus \{r\}$. The revenue of ancestors cannot exceed the payment from the buyer in any trade.
\end{enumerate}
\end{definition}

The PRM-aware model is specified by $\Model = (\calT, \pi, \calF)$.
A distinguished PRM is the \emph{baseline mechanism}, which does not perform profit reallocation:

\begin{definition}[Baseline Mechanism]
\label{def:baseline}
The \emph{baseline mechanism} $\pi^0$ allocates the entire payment to the direct seller: $\pi^0(i, \Parent(i)) = 1$ and $\pi^0(i, j) = 0$ for $j \neq \Parent(i)$.
\end{definition}

\subsection{Utility and Solution Concept}
\label{subsec:solution}

Given a PRM $\pi$ and terminal history $h$, the realized utility of player $i$ is:
\begin{equation}
\label{eq:utility-raw}
U_i(h; v_i) = y_i \cdot (v_i - p_i) + \sum_{j \in \Desc(i)} y_j \cdot \pi(j,i) \cdot p_j.
\end{equation}
The first term is player $i$'s own surplus from trade $i$; the second term aggregates the profit reallocation routed to $i$ from descendant trades, which by the Path Only condition (\cref{def:prm}) arrive only from $\Desc(i)$. We also denote social welfare as the aggregate utility $\SW(h;\bmv) = \sum_{i \in \calI} U_i(h;v_i)$.

\begin{figure*}[t]
\centering
\resizebox{0.9\textwidth}{!}{%
\begin{tikzpicture}[
  >={Latex[length=2mm]},
  rootN/.style={circle, draw, minimum size=8mm, inner sep=0pt, fill=croot,   text=white, font=\small},
  playerN/.style={circle, draw, minimum size=8mm, inner sep=0pt, fill=cplayer, text=white, font=\small},
  leafN/.style={circle, draw, minimum size=8mm, inner sep=0pt, fill=cleaf,   text=white, font=\small},
  db/.style={cylinder, shape border rotate=90, draw, aspect=0.28,
             minimum width=6.5mm, minimum height=5.5mm, fill=cdb, font=\tiny, inner sep=1pt},
  dyn/.style={font=\footnotesize, text=cdyn, align=left, inner sep=1pt},
  val/.style={font=\footnotesize, inner sep=1pt},
  ubox/.style={draw, rounded corners=2pt, font=\footnotesize, inner sep=3.5pt},
  ed/.style={->, line width=0.6pt, black!70},
  reOne/.style={->, line width=1pt, red!80!black},
  reHalf/.style={->, line width=1pt, creall},
  dcon/.style={-{Latex[length=1.6mm]}, dotted, line width=0.7pt, black!55},
]
\node[rootN]   (r)  at (0,0)        {$r$};
\node[playerN] (l)  at (-3.4,-2.1)  {$l$};
\node[playerN] (u)  at (3.4,-2.1)   {$u$};
\node[leafN]   (ll) at (-5.1,-4.3)  {$ll$};
\node[leafN]   (lu) at (-2.1,-4.3)  {$lu$};
\node[leafN]   (ul) at (2.1,-4.3)   {$ul$};
\node[leafN]   (uu) at (5.1,-4.3)   {$uu$};
\node[db, left=1.5mm of r]  (dbr) {data};
\node[db, left=1.5mm of u]  (dbu) {data};
\node[db, left=1.5mm of uu] {data};
\draw[ed] (r)--(l); \draw[ed] (r)--(u);
\draw[ed] (l)--(ll); \draw[ed] (l)--(lu);
\draw[ed] (u)--(ul); \draw[ed] (u)--(uu);
\node[val, below=0.5mm of r]  {$v_r=0$};
\node[val, anchor=east] at ([xshift=-1mm]l.west)   {$v_l=2$};
\node[val, anchor=east] at ([xshift=-1mm]dbu.west) {$v_u=-1$};
\node[val, below=0.5mm of ll] {$v_{ll}=3$};
\node[val, below=0.5mm of lu] {$v_{lu}=4$};
\node[val, below=0.5mm of ul] {$v_{ul}=5$};
\node[val, below=0.5mm of uu] {$v_{uu}=6$};
\node[dyn, right=1mm of l]  {$p_l=2$\\$x_l=0$};
\node[dyn, right=1mm of u]  {$p_u=1$\\$x_u=1$};
\node[dyn, right=1mm of ll] {$p_{ll}=6$\\$x_{ll}=0$};
\node[dyn, right=1mm of lu] {$p_{lu}=5$\\$x_{lu}=1$};
\node[dyn, right=1mm of ul] {$p_{ul}=4$\\$x_{ul}=0$};
\node[dyn, right=1mm of uu] {$p_{uu}=5$\\$x_{uu}=1$};
\coordinate (os) at (uu.north);
\draw[reOne]  (u.north west) to[out=150,in=-32]
      node[pos=0.5, right=0.6mm, font=\scriptsize, text=red!80!black] {$1$} (r.east);
\draw[reHalf] (os) to[out=122,in=-42]
      node[pos=0.42, left=0.4mm, font=\scriptsize, text=creall] {$0.5$} (u.south east);
\draw[reHalf] (os) to[out=76,in=-14]
      node[pos=0.10, right=0.4mm, font=\scriptsize, text=creall] {$0.5$} (r.east);
\node[ubox, anchor=west] (Ur)  at (7.6,0.1)   {$U_r={\color{black}0}+{\color{red!80!black}1}\!\cdot\!{\color{cdyn}1}+{\color{creall}0.5}\!\cdot\!{\color{cdyn}5}=3.5$};
\node[ubox, anchor=west] (Uu)  at (7.6,-2.4)  {$U_u={\color{black}-1}-{\color{cdyn}1}+{\color{creall}0.5}\!\cdot\!{\color{cdyn}5}=0.5$};
\node[ubox, anchor=west] (Uuu) at (7.6,-4.3)  {$U_{uu}={\color{black}6}-{\color{cdyn}5}=1$};
\draw[dcon] (Ur.west)  to[bend right=8] (r.east);
\draw[dcon] (Uu.west)  -- (u.east);
\draw[dcon] (Uuu.west) -- (uu.east);
\draw[decorate, decoration={brace, amplitude=5pt, mirror}, black!70]
      ([yshift=-6mm]ll.south) -- ([yshift=-6mm]lu.south);
\node[ubox, below=13mm of $(ll)!0.5!(lu)$] {$U_l=U_{ll}=U_{lu}=0$};
\node[ubox, below=9mm of ul] {$U_{ul}=0$};
\begin{scope}[shift={(-5.3,-6.65)}]
  \node[rootN, scale=0.8] (lgp) at (0,0) {$r$};
  \node[right=1mm of lgp, font=\footnotesize] (lgpt) {: player};
  \node[right=6mm of lgpt, font=\footnotesize] (lgv) {$v_r$: private valuation};
  \node[right=6mm of lgv, font=\footnotesize, text=cdyn] (lgd) {$p_u,x_u$};
  \node[right=1mm of lgd, font=\footnotesize] (lgdt) {: game dynamic};
  \coordinate (a0) at ([xshift=7mm]lgdt.east);
  \draw[reHalf] (a0) -- ($(a0)+(0.8,0)$) node[midway, above=0.3mm, font=\scriptsize, text=creall] {$0.5$};
  \coordinate (a1) at ($(a0)+(1.25,0)$);
  \draw[reOne] (a1) -- ($(a1)+(0.8,0)$) node[midway, above=0.3mm, font=\scriptsize, text=red!80!black] {$1$};
  \node[anchor=west, font=\footnotesize] at ($(a1)+(0.95,0)$) {: profit reallocation};
\end{scope}
\end{tikzpicture}%
}
\caption{A running example of the tree-based data trading game. Best viewed in color.}
\label{fig:example}
\end{figure*}
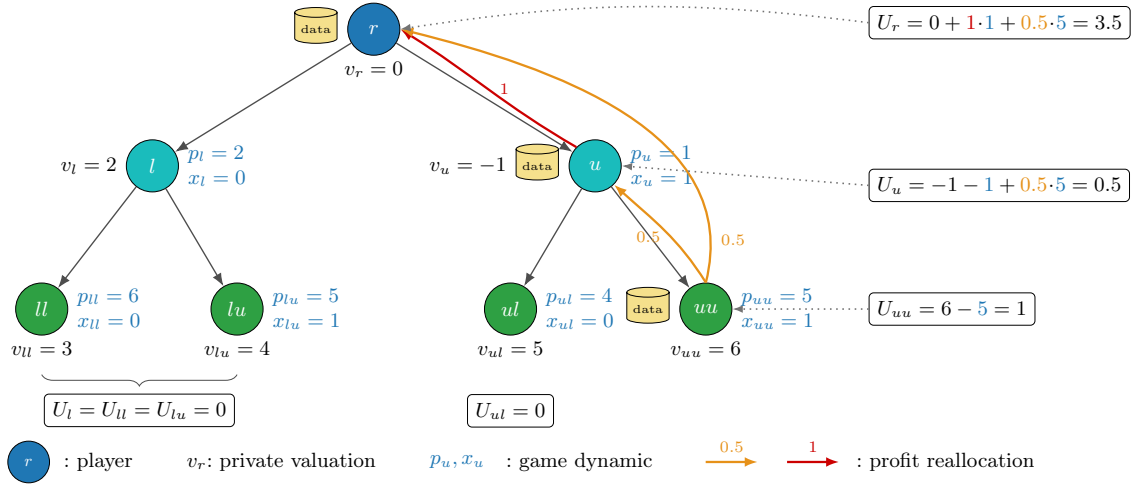

\begin{example}
\label{ex:running}
\cref{fig:example} shows an instance with $N = 7$ players; the valuations, the PRM $\pi$, and all actions are given there, and the realized utilities follow from \cref{eq:utility-raw}. Two effects are worth highlighting. First, player $u$ has a \emph{negative} value $v_u = -1$, so its direct trade loses $v_u - p_u = -2$; yet the reallocated profit $\pi(uu,u)\,p_{uu} = 2.5$ from the successful downstream trade $uu$ leaves it with $U_u = 0.5$---so it's sometimes rational to buy the data even if the valuation is negative. Second, trade $lu$ fails despite $x_{lu}=1$, because the upstream trade $l$ fails ($x_l=0$), so $l, lu, ll$ all earn zero.
\end{example}

Under the strategy profile $(\tbmp, \tbmx)$, the expected utility of player $i$ is:
\begin{equation}
\label{eq:expected-utility}
u_i(\tbmp, \tbmx) \;\coloneqq\; \bbE_{\bmv \sim \calF}\!\big[U_i(h(\tbmp, \tbmx; \bmv); v_i)\big],
\end{equation}
where $h(\tbmp, \tbmx; \bmv)$ is the terminal history generated by $(\tbmp, \tbmx)$ and $\bmv$ according to game dynamic (\cref{def:game-evolve}).

How would strategic players behave under a given PRM? As the game is sequential and players hold private information, we adopt \emph{perfect Bayesian equilibrium} (PBE)~\citep{fudenberg1991game,MWG}---the standard refinement of Nash equilibrium for such games---as our solution concept: each player must act optimally at \emph{every} information set she may reach, under a belief about the others' valuations that is updated by Bayes' rule.

Player $i$'s \emph{information set} is the pair $(h_i, v_i)$ of her observed history and private valuation. A \emph{belief system} $\mu = \{\mu_i\}_{i \in \calI}$ specifies, at each information set, a distribution $\mu_i(\cdot \mid h_i, v_i) \in \Delta(\calV_{-i})$ over the other players' valuations $\bmv_{-i}$, where $\Delta(\calV_{-i})$ is the set of probability distributions on $\calV_{-i}$; this is player $i$'s posterior about $\bmv_{-i}$ upon observing $(h_i, v_i)$. Given a strategy profile $(\tbmp, \tbmx)$ and the belief $\mu_i$, player $i$'s conditional expected utility at $(h_i, v_i)$ is
\begin{equation}
\label{eq:cond-utility}
\bbU_i\big((\tbmp, \tbmx), \mu_i \mid h_i, v_i\big) \;\coloneqq\; \bbE_{\bmv_{-i} \sim \mu_i(\cdot \mid h_i, v_i)}\!\big[\, U_i\big(h(\tbmp, \tbmx;\, (v_i, \bmv_{-i})); v_i\big)\big],
\end{equation}
where $U_i$ is the realized utility~\eqref{eq:utility-raw} and $h(\cdot\,; \bmv)$ the terminal history generated by the game dynamics (\cref{def:game-evolve}).

\begin{definition}[Perfect Bayesian Equilibrium~\citep{fudenberg1991game,MWG}]
\label{def:pbe}
An assessment $\big((\tbmp^*, \tbmx^*), \tbmmu\big)$ of a strategy profile and a belief system is a \emph{perfect Bayesian equilibrium} if both of the following hold.
\begin{enumerate}[label=(\alph*),left=0em]
\item \textup{(Sequential rationality)} For every player $i \in \calI$, every information set $(h_i, v_i)$, and every deviation $(\tp_i, \tx_i)$ in player $i$'s own strategy components ($\tp_i$ only for the root, $\tx_i$ only for a leaf, both otherwise),
\[
\bbU_i\big((\tbmp^*, \tbmx^*), \tmu_i \mid h_i, v_i\big) \;\ge\; \bbU_i\big((\tp_i, \tbmp^*_{-i}, \tx_i, \tbmx^*_{-i}), \tmu_i \mid h_i, v_i\big).
\]
\item \textup{(Bayes-consistency)} At every information set $(h_i, v_i)$ reached with positive probability under $(\tbmp^*, \tbmx^*)$, the belief $\tmu_i(\cdot \mid h_i, v_i)$ equals the posterior distribution of $\bmv_{-i}$ induced by the prior $\calF$ and the profile $(\tbmp^*, \tbmx^*)$ given $(h_i, v_i)$.
\end{enumerate}
\end{definition}

\section{Equilibrium Computation}
\label{sec:algo}

Computing an equilibrium is far from trivial: a player's strategy maps the exponentially large space of histories to multi-dimensional actions. Our structural reduction (\cref{subsec:simplification}) shows that, under the Markov property, equilibrium strategies collapse to history-independent pricing and threshold rules and a seller's optimal pricing strategy separates across its children's subtrees. Motivated by this, we further derive an exact polynomial-time algorithm for finitely supported valuations (\cref{subsec:exact-finite}) and an FPTAS for continuous valuations (\cref{subsec:approx-continuous}).
For space limits, we defer all proofs to \cref{sec:omitted-proofs}.

\subsection{Equilibrium Simplification}
\label{subsec:simplification}

First, we show that the Markovian structure on trees yields a clean posterior, which provides a foundation for further strategy simplification.

\begin{fact}
\label{fact:simplify}
Let $\calF$ satisfy \cref{asp:markov}. For any strategy profile $(\tbmp, \tbmx)$ and any player $i$, the posterior distribution of descendant valuations given $(v_i, h_i)$ satisfies $\calF(\bmv_{\Desc(i)} \mid v_i, h_i) = \calF(\bmv_{\Desc(i)} \mid v_i)$.
\end{fact}

Intuitively, $h_i = \phi_i(\bmv_{\Path(i)})$ is a deterministic function $\phi_i$ of the ancestor valuations $\bmv_{\Path(i)}$, which are conditionally independent of $\bmv_{\Desc(i)}$ given $v_i$ under \cref{asp:markov}.
\cref{fact:simplify} allows us to set $\tmu_i(\bmv_{\Desc(i)}|h_i,v_i) = \calF(\bmv_{\Desc(i)} \mid v_i)$ without loss of generality.
We now present the core simplification.

\begin{restatable}{lemma}{lemSimplifySell}
\label{lem:simplify-sell}
Let $\calF$ satisfy \cref{asp:markov}.
There exists a perfect Bayesian equilibrium $(\tbmp, \tbmx, \tbmmu)$ and functions $\{\tp^\dagger_j(v)\}_{j \in \calI \setminus \{r\}}$, $\{\tx^\dagger_i(v,p)\}_{i \in \calI \setminus \{r\}}$ such that:
\begin{equation}
\label{eq:lem:simplify-sell}
\begin{aligned}
\tp_i(h_i, j, v_i) &= \tp^\dagger_j(v_i), &&\forall i \notin \calL,\; j \in \Children(i),\; \forall v_i, h_i \\
\tx_i(h_i, v_i) &= \tx^\dagger_i(v_i, p_i), &&\forall i \neq r,\; \forall v_i, h_i
\end{aligned}
\end{equation}
\end{restatable}

\cref{lem:simplify-sell} establishes that, despite the rich history space, equilibrium pricing for each child depends only on the seller's valuation and the buying decision only on the buyer's valuation and the offered price: by \cref{fact:simplify}, $h_i$ carries no information about descendant valuations beyond $v_i$, so player $i$ does not benefit from more knowledge on $h_i$.
The proof is done by backward induction on the tree.

\begin{restatable}{corollary}{corSimplifyBuy}
\label{cor:simplify-buy}
Under \cref{asp:markov}, there exists a perfect Bayesian equilibrium $(\tbmp, \tbmx, \tbmmu)$ satisfying the conclusions of \cref{lem:simplify-sell} in which each buyer $i \in \calI \setminus \{r\}$ uses a threshold strategy: there exists a function $\ts^\dagger_i: \calV_i \to \bbR$ such that
\begin{equation}
\label{eq:threshold}
\tx^\dagger_i(v_i, p_i) =
\begin{cases}
1, & \text{if } \ts^\dagger_i(v_i) \ge p_i, \\
0, & \text{otherwise,}
\end{cases}
\end{equation}
where $\ts^\dagger_i(v_i)$ represents the maximum price player $i$ is willing to pay.
\end{restatable}

By \cref{lem:simplify-sell,cor:simplify-buy}, a perfect Bayesian equilibrium can be characterized by a \emph{simplified strategy profile} $(\tbmp, \tbms) = (\{\tp_j\}_{j \in \calI \setminus \{r\}},\; \{\ts_i\}_{i \in \calI \setminus \{r\}})$, where $\tp_j: \calV_{\Parent(j)} \to \bbR_+$ is the pricing function and $\ts_i: \calV_i \to \bbR$ is the threshold function. We denote $(\tbmp^*, \tbms^*)$ as an equilibrium, if, by transforming $(\tbmp^*, \tbms^*)$ into original strategy profile $(\tbmp^*, \tbmx^*)$ as in \cref{eq:lem:simplify-sell,eq:threshold}, there is $\tbmmu^*$ such that $((\tbmp^*, \tbmx^*), \tbmmu^*)$ is a perfect Bayesian Equilibrium as in \cref{def:pbe}.

The key structural property of the tree model is the following subtree decomposition. For each non-leaf player $i \in \calI \setminus \calL$, each child $j \in \Children(i)$, and a simplified strategy profile $(\tbmp, \tbms)$, we define the \emph{gain function} for player $i$ selling to child $j$ at price $p_j$:
\begin{equation}
\label{eq:gain-child}
G^{(j)}_i(v_i, p_j) \;\coloneqq\; \bbE_{\bmv_{\Subtree(j)} \sim \calF(\cdot|v_i)} \bigg[ x_j \cdot \bigg(\pi(j,i)\cdot p_j + \sum_{m \in \Desc(j)} \pi(m,i) \cdot p_m \cdot y_{m \mid j} \bigg)\bigg],
\end{equation}
where, given the profile $(\tbmp, \tbms)$, $x_j = \one\{\ts_j(v_j) \ge p_j\}$, $p_m = \tp_m(v_{\Parent(m)})$, and $y_{m|j}$ (\cref{eq:y-conditional}) are determined within $\Subtree(j)$. Note that $G^{(j)}_i$ only depends on $(\tbmp, \tbms)$ restricted to $\Subtree(j)$; this is because the Markov property (\cref{asp:markov}) makes sibling subtrees conditionally independent given $v_i$.

\begin{restatable}{lemma}{lemSubtreeDecomp}
\label{lem:subtree-decomp}
Under \cref{asp:markov}, a simplified strategy profile $(\tbmp^*, \tbms^*)$ is a perfect Bayesian equilibrium if, for every non-leaf player $i \in \calI \setminus \calL$ and every valuation $v_i \in \calV_i$, the following two conditions hold:
\begin{enumerate}[left=0em]
\item \textup{(Pricing decomposition)} The price offered to each child $j \in \Children(i)$ maximizes the corresponding gain:
\begin{equation}
\label{eq:pricing-decomp}
\tp^*_j(v_i) \in \argmax_{p_j \ge 0}\; G^{(j)}_i(v_i, p_j), \quad \forall j \in \Children(i).
\end{equation}
\item \textup{(Threshold decomposition)} The threshold function decomposes as:
\begin{equation}
\label{eq:threshold-decomp}
\ts^*_i(v_i) = v_i + \sum_{j \in \Children(i)} \max_{p_j \ge 0}\; G^{(j)}_i(v_i, p_j).
\end{equation}
\end{enumerate}
\end{restatable}

We write the two conditions~\eqref{eq:pricing-decomp}--\eqref{eq:threshold-decomp} compactly as
\begin{equation}
\label{SE:transform}
\tp^*_j(v_i) \in \argmax_{p_j \ge 0}\; G^{(j)}_i(v_i, p_j), \qquad \ts^*_i(v_i) = v_i + \sum_{j \in \Children(i)} \max_{p_j \ge 0}\; G^{(j)}_i(v_i, p_j),
\end{equation}
for all $i \in \calI \setminus \calL$ and $j \in \Children(i)$.

The proof, deferred to \cref{sec:omitted-proofs}, decomposes player $i$'s conditional expected utility into the additive per-child gains $\{G^{(j)}_i\}_{j \in \Children(i)}$ and shows that conditions~\eqref{eq:pricing-decomp}--\eqref{eq:threshold-decomp} make the prescribed actions maximize player $i$'s conditional expected utility at every information set.
Therefore, computing a perfect Bayesian equilibrium reduces to finding a simplified strategy profile $(\tbmp^*, \tbms^*)$ satisfying~\eqref{SE:transform}.

\subsection{Polynomial Algorithm for Finite Valuation Support}
\label{subsec:exact-finite}

When valuations take finitely many values, strategies admit finite tabular representations and we show that exact equilibria can be computed by a nested dynamic program.

\begin{assumption}[Finite Valuation Support]
\label{asp:finite}
$\calF$ has finite valuation support if $\Pr[\forall i \in \calI,\; v_i \in W] = 1$ for $\bmv \sim \calF$, where $W = \{w_1, \ldots, w_S\} \subseteq \bbR$.
\end{assumption}

Under \cref{asp:finite}, strategies have finite representations: $p^*(i,k) \coloneqq \tp^*_i(w_k)$ and $s^*(i,k) \coloneqq \ts^*_i(w_k)$ for $i \in \calI \setminus \{r\},\; k \in [S]$. The conditional distribution is represented by $Q = \{q_{i,k,\ell}\}$ where $q_{i,k,\ell} \coloneqq f_i(v_i = w_\ell \mid v_{\Parent(i)} = w_k)$. We denote the finite-support model as $(\calT, \pi, Q)$.

\begin{restatable}{theorem}{thmAlgoExact}
\label{thm:algo-exact}
Let \cref{asp:markov,asp:finite} hold. Given the model $(\calT, \pi, Q)$, there exists an algorithm (\cref{alg:exact}) that computes an exact perfect Bayesian equilibrium in $O\big(N d\, S^4\big)$ time, where $S = |W|$ is the support size and $d \coloneqq \max_{i \in \calI} |\Path(i)|$ is the depth of $\calT$.
\end{restatable}

\noindent\emph{Proof Sketch.}
\cref{alg:exact} computes the simplified equilibrium by backward induction over the tree, from the leaves to the root. Suppose the thresholds and prices of every node in $\Subtree(j)$ are already known for each child $j$ of a node $i$. By the subtree decomposition (\cref{lem:subtree-decomp}), player $i$'s problem separates across children, so its price to child $j$ maximizes the per-child gain $G^{(j)}_i(w_k,\cdot)$ independently. We further observe that the optimal price need only be searched over the finite candidate set $\{s^*(j,\ell):\ell\in[S]\}$ of buyer thresholds: between two consecutive thresholds the set of accepting buyers is fixed while raising the price weakly increases revenue, so an optimum can be attained at some threshold. Hence $i$'s strategy follows once $G^{(j)}_i(w_k,p)$ is evaluated for the $O(S)$ candidate prices and $O(S)$ valuations.

The core difficulty is evaluating $G^{(j)}_i(w_k,p)$: by~\eqref{eq:gain-child} it is an expectation of reallocated payments over the \emph{entire} profile $\bmv_{\Subtree(j)}$, and summing naively over all profiles is exponential in $|\Subtree(j)|$. We overcome this with a forward dynamic program (\cref{alg:forward-dp}) that decomposes the expectation into a weighted sum of the $O(|\Subtree(j)|\,S)$ elementary events $Y(m,\ell) = \{v_m = w_\ell \text{ and every player from } j \text{ down to } m \text{ is willing to buy}\}$, whose probabilities propagate top-down through $\Subtree(j)$ via the Markov transition $Q$ in $O(|\Subtree(j)|\,S^2)$ time. Repeating over the $O(S)$ valuations and $O(S)$ candidate prices of $i$ gives $O(|\Subtree(j)|\,S^4)$ per child; summing over all edges and using $\sum_{j\neq r}|\Subtree(j)| = \sum_m \mathrm{depth}(m) \le Nd$ yields the $O(Nd\,S^4)$ bound.\qed

\vspace{1em}

\cref{alg:exact} makes PRMs evaluable by computing and evaluating the exact equilibria under PRMs in polynomial time.
In \cref{sec:exp}, we further use \cref{alg:exact} for empirical evaluation of PRMs.


\begin{algorithm}[t]
\caption{Exact Perfect Bayesian Equilibrium on Trees}
\label{alg:exact}
\KwIn{Tree $\calT$ with node set $\calI$, PRM $\pi$, transition probabilities $Q = \{q_{i,k,\ell}\}$, support $W = \{w_1, \ldots, w_S\}$}
\KwOut{Equilibrium strategies $\{s^*(i,k)\}_{i \in \calI \setminus \{r\},\, k \in [S]}$, $\{p^*(j,k)\}_{j \in \calI \setminus \{r\},\, k \in [S]}$}
\For{each leaf $i \in \calL$, $k \in [S]$}{
  $s^*(i,k) \gets w_k$\;
}
\For{each non-leaf $i \in \calI \setminus \calL$ in bottom-up order}{
  \For{$k = 1, \ldots, S$}{
    \For{each $j \in \Children(i)$}{
      $\calP_j \gets \{\ell \in [S] : s^*(j, \ell) \ge 0\}$\;
      \eIf{$\calP_j = \varnothing$}{
        $p^*(j, k) \gets 0$;\quad $G^*_j \gets 0$\;
      }{
        \For{each $k' \in \calP_j$}{
          $G(k') \gets \textsc{ForwardDP}(\calT, \pi, Q, j, i, k, s^*(j, k'))$\;
        }
        $l^* \gets \argmax_{k' \in \calP_j} G(k')$\;
        $p^*(j, k) \gets s^*(j, l^*)$;\quad $G^*_j \gets G(l^*)$\;
      }
    }
    $s^*(i, k) \gets w_k + \sum_{j \in \Children(i)} G^*_j$\;
  }
}
\textbf{return}: $\{s^*(i,k)\}_{i \in \calI \setminus \{r\},\, k \in [S]}$, $\{p^*(j,k)\}_{j \in \calI \setminus \{r\},\, k \in [S]}$
\end{algorithm}

\subsection{FPTAS for Continuous Valuation Support}
\label{subsec:approx-continuous}

With continuous support, exact equilibria generally lack a finite representation, so we seek an $\varepsilon$-approximate equilibrium---relaxing the conditions of \cref{SE:transform} by a tolerance $\varepsilon$.

\begin{definition}[$\varepsilon$-Approximate Equilibrium]
\label{def:eps-equilibrium}
Given a profile $(\tbmp, \tbms)$, let $\ts^\dagger_i(v_i) \coloneqq v_i + \sum_{j \in \Children(i)} \max_{p_j \ge 0} G^{(j)}_i(v_i, p_j)$ be the \emph{threshold} of player $i$ against the others' strategies in $(\tbmp, \tbms)$. The profile $(\tbmp, \tbms)$ is an \emph{$\varepsilon$-approximate equilibrium} if, for every player $i$ and every $v_i \in \calV_i$, both of the following hold:
\begin{enumerate}[label=(\alph*),left=0em]
\item \textup{(Pricing near-optimality)} For all $j \in \Children(i)$,
\[
G^{(j)}_i(v_i, \tp_j(v_i)) - \max_{p_j \ge 0} G^{(j)}_i(v_i, p_j) \;\ge\; -\varepsilon \cdot \max\{1, \ts^\dagger_i(v_i)\}.
\]
\item \textup{(Threshold accuracy)} $\;|\ts_i(v_i) - \ts^\dagger_i(v_i)| \le \varepsilon \cdot \max\{1, \ts^\dagger_i(v_i)\}$.
\end{enumerate}
At $\varepsilon = 0$ both conditions reduce to~\eqref{SE:transform}, recovering an exact perfect Bayesian equilibrium (\cref{lem:subtree-decomp}); $\varepsilon$ thus measures how far each player is from best-responding.
\end{definition}

\begin{definition}[Log-Lipschitzness]
\label{asp:lipschitz}
Under \cref{asp:markov} with each valuation supported on $[L_0, H_0]$, we say $f_i$ is $K_0$-log-Lipschitz if, for each $i \in \calI \setminus \{r\}$ and all $v_i, a, b \in [L_0, H_0]$:
$|\log f_i(v_i \mid a) - \log f_i(v_i \mid b)| \le K_0\,|a - b|$.
\end{definition}

\begin{restatable}{theorem}{thmAlgoApprox}
\label{thm:algo-approx}
Under \cref{asp:markov}, assume $f_i$ is $K_0$-log-Lipschitz for all $i \in \calI \setminus \{r\}$, and each $v_i$ has support $[L_0, H_0] \subseteq \bbR$. Then there exists an algorithm (\cref{alg:approx}) that outputs an $\varepsilon$-approximate equilibrium in $\poly(N, \varepsilon^{-1}, H_0 - L_0, K_0)$ time.
\end{restatable}

The algorithm (\cref{alg:approx}, deferred to \cref{sec:omitted-proofs}) discretizes the continuous supports into a finite grid, forming a discretized model with distribution $\hat\calF$, and computes its exact equilibrium via \cref{alg:exact}. The proof, also in \cref{sec:omitted-proofs}, shows that exact equilibria of the discretized model are approximate equilibria of the original continuous model, using the log-Lipschitz condition to bound the discretization error.
The result indicates that the discretized model is a provably faithful approximation for the continuous one.
\section{Equilibrium Comparison between PRMs}
\label{sec:theoretic}

This section establishes that profit reallocation mechanisms on trees expand equilibrium trade events compared with the baseline mechanism. As there may be multiple equilibria, we fix the simplified perfect Bayesian equilibrium $(\tbmp^*, \tbms^*)$ satisfying~\eqref{SE:transform}, constructed by the backward induction of \cref{lem:simplify-sell,cor:simplify-buy}.
Ties in each $\argmax$ are resolved by choosing the \emph{smallest} element. This tie-breaking rule is independent of $\pi$, thus does not provide advantages to some specific PRM. We impose two additional assumptions:

\begin{assumption}[Positivity]
\label{asp:posi}
$\Pr_{\bmv \sim \calF}[v_i > 0] = 1$ for all $i \in \calI$.
\end{assumption}

\begin{assumption}[Homogeneity]
\label{asp:homo}
Under \cref{asp:markov}, for all $\alpha, w, v > 0$ and $i \in \calI \setminus \{r\}$: $F_i(\alpha v \mid \alpha w) = F_i(v \mid w)$.
\end{assumption}

Positivity says that a data product is always worth more than the cost of producing it. Homogeneity is a scaling property common in economic modeling~\citep[\S6.2]{AGT:nisan_algorithmic_2007}, \citep[\S3.2.3]{DSGE}; it holds whenever the value of the output is proportional to the value of the input.

Given a simplified strategy profile $(\tbmp, \tbms)$, we define the \emph{trade event} for trade $j \in \calI \setminus \{r\}$ as
\begin{equation}
\label{eq:trade-event}
E_j(\tbmp, \tbms) \;\coloneqq\; \big\{\bmv \in \calV : \ts_j(v_j) \ge \tp_j(v_{\Parent(j)})\big\},
\end{equation}
\ie, the set of valuation profiles under which player $j$ is willing to buy at the price her parent offers.


\subsection{Local Comparison between Mechanisms}
\label{subsec:partial}

We first compare two PRMs that reallocate profit differently along a single trade, where one rewards the upstream seller more from downstream resale than the other. This relation is inherently asymmetric, so we phrase it as \emph{local dominance} rather than a symmetric difference.

\begin{definition}[Local Dominance]
\label{def:local-diff}
For a non-root player $j \in \calI \setminus \{r\}$, let $i \coloneqq \Parent(j)$. A PRM $\pi^1$ \emph{locally dominates} another PRM $\pi^2$ \emph{at trade $j$} if:
\begin{enumerate}[label=(\alph*),left=0em]
\item $\pi^1(m, k) = \pi^2(m, k)$ for all $m \in \calI \setminus \{r\}$ and all $k \in \calI \setminus \{i\}$.
\item $\pi^1(j, i) \le \pi^2(j, i)$: player $i$'s one-shot profit from trade $j$ is weakly smaller under $\pi^1$.
\item $\pi^1(m, i) \ge \pi^2(m, i)$ for all $m \in \Desc(j)$: player $i$'s reallocated profit from downstream trades is weakly larger under $\pi^1$.
\item $\pi^1(m, k) = \pi^2(m, k) = \pi^0(m, k)$ for all $k \in \Path(i)$ and all $m \in \calI \setminus \{r\}$: both mechanisms use baseline mechanisms for all strict ancestors of $i$.
\item $\pi^1(m, i) = \pi^2(m, i)$ for all $m \notin \Subtree(j)$: player $i$'s reallocation is unchanged outside $\Subtree(j)$.
\end{enumerate}
That is, relative to $\pi^2$, the dominating mechanism $\pi^1$ shifts $i$'s reward away from its one-shot payment on trade $j$ toward the downstream resale within $\Subtree(j)$, leaving everything else fixed.
\end{definition}

\begin{restatable}{theorem}{thmPartial}
\label{thm:partial}
Let $\calF$ satisfy \cref{asp:markov,asp:posi,asp:homo}. Let $\pi^1$ locally dominate $\pi^2$ at trade $j$ (\cref{def:local-diff}). Let $(\tbmp^{k,*}, \tbms^{k,*})$ be the equilibrium under $\pi^k$ for $k \in \{1,2\}$. Then:
\begin{enumerate}
\item $E_m(\tbmp^{1,*}, \tbms^{1,*}) = E_m(\tbmp^{2,*}, \tbms^{2,*})$ for all $m \in \calI \setminus \{r\}$ with $m \neq j$.
\item $E_j(\tbmp^{2,*}, \tbms^{2,*}) \subseteq E_j(\tbmp^{1,*}, \tbms^{1,*})$.
\end{enumerate}
\end{restatable}

\emph{Proof Sketch.}
Write $i \coloneqq \Parent(j)$. The two mechanisms differ only in how they reward $i$, and only from within $\Subtree(j)$ (conditions (a)--(e)); every other reallocation is fixed. We show each trade event $E_m$ with $m\neq j$ is unchanged, then analyze trade $j$.

\emph{(i) $m\notin\Subtree(i)\cup\Path(i)$ and (ii) $m\in\Desc(j)$.} These equilibria are computed inside subtrees whose \emph{internal} profit reallocations are same between $\pi_1$ and $\pi_2$, so their equilibrium strategies, hence $E_m$, are identical.

\emph{(iii) Sibling trades $m\in\Subtree(c)$, $c\in\Children(i)\setminus\{j\}$.} By subtree decomposition (\cref{lem:subtree-decomp}), $i$'s optimal price to $c$ and equilibrium strategy within $\Subtree(c)$ depend only on the reallocations within $\Subtree(c)$ as well as the reallocations from $\Subtree(c)$ to $i$. While condition (e) ensures that these reallocations are same between $\pi_1$ and $\pi_2$. Similarly, $E_m$, are identical.

\emph{(iv) Path trades $m\in\Path(i)\cup\{i\}$.} Both $\pi_1$ and $\pi_2$ use the baseline mechanism along this path (condition (d)), so the seller collects only her \emph{one-shot} payment. We first derive a homogeneity scaling lemma (\cref{lem:scaling} in \cref{sec:omitted-proofs}): $\ts_m(v_m)$ is linear in $v_m$, and hence $E_m = \{\bmv : v_m/v_{\Parent(m)} \ge \tau_m\}$ with $\tau_m$ depending only on $F_m$. Therefore, $E_m$ is identical between $\pi_1$ and $\pi_2$.

\emph{Trade $j$.} By (ii), all strategies within $\Subtree(j)$ are unchanged, so $\ts_j$ is identical between $\pi_1$ and $\pi_2$. Only the price $\tp_j$ may differ. Split $i$'s gain $G^{(j)}_i = r_i + d_i$ into the one-shot revenue $r_i$ from $j$ and the reallocated revenue $d_i$ from $\Desc(j)$. Under the dominating $\pi^1$, $i$'s one-shot revenue is smaller (condition (b)) and its reallocated revenue is larger (condition (c)): $i$ earns \emph{less} directly from $j$ and \emph{more} from what $j$ resells. As the reallocated revenue accrues only when $j$ buys and then resells, $i$ is incentivized to induce more downstream trade, and its only lever is to \emph{lower} $\tp_j$, admitting more willingness for player $j$ to buy, and resulting $E_j(\tbmp^{2,*}, \tbms^{2,*}) \subseteq E_j(\tbmp^{1,*}, \tbms^{1,*})$.
\qed

\subsection{Comparison with Baseline Mechanism}
\label{subsec:global}

The local comparison result (\cref{thm:partial}) serves as the building block for a global comparison between any PRM and the baseline mechanism.

\begin{restatable}{theorem}{thmGlobal}
\label{thm:global}
Let \cref{asp:markov,asp:posi,asp:homo} hold.
Let $\pi^0$ be the baseline mechanism (\cref{def:baseline}) and $\pi$ be any feasible PRM satisfying \cref{def:prm}. Denote by $(\tbmp^{0,*}, \tbms^{0,*})$ and $(\tbmp^*, \tbms^*)$ the equilibria under $\pi^0$ and $\pi$, respectively. Then $E_j(\tbmp^{0,*}, \tbms^{0,*}) \subseteq E_j(\tbmp^*, \tbms^*)$ for all $j \in \calI \setminus \{r\}$.
\end{restatable}

\begin{proof}
We interpolate from $\pi^0$ to $\pi$ through a sequence of mechanisms in which each mechanism locally dominates its predecessor at a \emph{single} trade in the sense of \cref{def:local-diff}, and then chain the trade-event inclusions of \cref{thm:partial}. The key point is how to construct a sequence of mechanisms that satisfying this condition.

We list the non-leaf players $\calI \setminus \calL$ as $i_1, \ldots, i_M$ in a bottom-up order, so that every player $i$ precedes all players in $\Path(i)$. For each $i_m$, enumerate its children $\Children(i_m) = \{j_{m,1}, \ldots, j_{m,d_m}\}$ in any fixed order.
We will define atomic operations indexed by the pairs $(m,t)$, $1 \le m \le M$, $1 \le t \le d_m$, and take each operation in lexicographic order starting with $\pi^0$; there are $\sum_m d_m = N-1$ number of such operation, with each operation corresponding to a trade.
Specifically, operation $(m,t)$ switches player $i_m$'s reallocation \emph{from the subtree} $\Subtree(j_{m,t})$ from $\pi^0$ to $\pi$, leaving everything else unchanged. Writing $\sigma\in\{1,...,N-1\}$ as the integer index of $(m,t)$, this defines mechanisms $\pi^{[0]} = \pi^0, \pi^{[1]}, \ldots, \pi^{[N-1]} = \pi$ by
\[
\pi^{[\sigma]}(n, k) \;=\;
\begin{cases}
\pi(n, i_m), & k = i_m \text{ and } n \in \Subtree(j_{m,t}),\\
\pi^{[\sigma-1]}(n, k), & \text{otherwise.}
\end{cases}
\]
Equivalently, for trade index $n \in \calI \backslash \{r\}$ and player index $k\in \calI \backslash \calL$ where $k \in \Path(n)$, let $c_k(n)$ denote the unique child of $k$ on the path from $k$ to $n$, \ie, the single element of $(\Path(n)\cup\{n\}) \cap \Children(k)$. Then $\pi^{[\sigma]}(n,k)$ equals the target $\pi(n,k)$ once the operation switching $k$'s reallocation from $\Subtree(c_k(n))$ (\ie, the operation indexed by $(m,t)$ where $i_m = k, j_{m,t} = c_k(n)$) has occurred, and equals the baseline $\pi^0(n,k)$ otherwise.
Let $(\tbmp^{[\sigma],*}, \tbms^{[\sigma],*})$ be the equilibrium under $\pi^{[\sigma]}$ and abbreviate $E^{[\sigma]}_n \coloneqq E_n(\tbmp^{[\sigma],*}, \tbms^{[\sigma],*})$.

Fix operation $(m,t)$ and write $i \coloneqq i_m$, $j \coloneqq j_{m,t}$, $\pi^- \coloneqq \pi^{[\sigma-1]}$, $\pi^+ \coloneqq \pi^{[\sigma]}$. By construction $\pi^-$ and $\pi^+$ agree everywhere except on $\{\pi(n,i) : n \in \Subtree(j)\}$. We verify that $\pi^+$ locally dominates $\pi^-$ at trade $j$ (\cref{def:local-diff}), \ie, $\pi^+$ plays the role of $\pi^1$ and $\pi^-$ that of $\pi^2$:
\begin{itemize}[left=0em]\itemsep0.1em
\item[(a)] only reallocation to $i$ changes, so $\pi^+(n,k) = \pi^-(n,k)$ for all $k \neq i$;
\item[(b)] the one-shot share $\pi(j,i)$ is set only at this operation, so $\pi^+(j,i) = \pi(j,i) \le 1 = \pi^0(j,i) = \pi^-(j,i)$;
\item[(c)] for $n \in \Desc(j)$, the share $\pi(n,i)$ is set only at this operation, so $\pi^+(n,i) = \pi(n,i) \ge 0 = \pi^0(n,i) = \pi^-(n,i)$ (here $\pi^0(n,i)=0$ because $i = \Parent(j) \neq \Parent(n)$);
\item[(d)] every strict ancestor $k \in \Path(i)$ follows $i$ in the bottom-up order, so no operation has yet altered reallocation \emph{to} $k$; hence $\pi^-(n,k) = \pi^+(n,k) = \pi^0(n,k)$ for all $n$ and all $k \in \Path(i)$;
\item[(e)] $\pi^-(n,i) = \pi^+(n,i)$ for all $n \notin \Subtree(j)$, since operation $(m,t)$ only changes the mechanism within $\Subtree(j)$.
\end{itemize}

We next show that $\pi^{[\sigma]}$ is a valid PRM (satisfying Path Only and Budget Feasibility) for every $\sigma$:
\begin{itemize}[left=0em]
\item[(1)] For Path Only, each element of $\pi^{[\sigma]}(n,k)$ is either $\pi(n,k)$ or $\pi^0(n,k)$. As Path Only holds for both $\pi(\cdot)$ and $\pi^0(\cdot)$, Path Only must hold for $\pi^{[\sigma]}(\cdot)$.
\item[(2)] For Budget Feasibility, fix a trade $n$ with ancestors $\Path(n) = \{k_1, \ldots, k_q\}$ ordered from $k_1 = \Parent(n)$ up to the root $k_q = r$. Since operations process trades in a bottom-up order, $\pi^{[\sigma]}(n,k_i)$ is switched from $\pi^0$ to $\pi$ before $\pi^{[\sigma]}(n,k_{i+1})$; hence at any stage the ``already switched ancestors'' are $K = \{k_1, \ldots, k_a\}$ or $K = \emptyset$.
We consider two subcase: (a) If $K = \emptyset$, then $\sum_b \pi^{[\sigma]}(n,k_b) = \sum_b \pi^0(n,k_b) = 1$. (b) If $K = \{k_1, \ldots, k_a\}$, notice that $\pi^0(n,k_i) = 0$ for $i > a \ge 1$, thus $\sum_b \pi^{[\sigma]}(n,k_b) = \sum_{b \le a} \pi^{[\sigma]}(n,k_b) = \sum_{b \le a} \pi(n,k_b) \le \sum_{b} \pi(n,k_b) \le 1$ by feasibility of $\pi$. 
\end{itemize}

Thus $\pi^+$ locally dominates $\pi^-$ at $j$, and \cref{thm:partial} yields $E^{[\sigma-1]}_j \subseteq E^{[\sigma]}_j$ (the dominated $\pi^-$ has the smaller trade event) and $E^{[\sigma-1]}_{n} = E^{[\sigma]}_{n}$ for every other trade $n$.
Composing the inclusions over all $N-1$ steps gives $E_j(\tbmp^{0,*}, \tbms^{0,*}) = E^{[0]}_j \subseteq E^{[N-1]}_j = E_j(\tbmp^*, \tbms^*)$ for all $j \in \calI \setminus \{r\}$.\qed
\end{proof}

The expansion of every trade event translates directly into a social welfare improvement.

\begin{corollary}[Welfare Improvement]
\label{cor:welfare}
Let \cref{asp:markov,asp:posi,asp:homo} hold and let $\pi$ be any \emph{budget-balanced} PRM, \ie, $\sum_{i \in \Path(j)} \pi(j,i) = 1$ for every $j \in \calI \setminus \{r\}$, with equilibrium $(\tbmp^*, \tbms^*)$; let $(\tbmp^{0,*}, \tbms^{0,*})$ be the equilibrium under the baseline mechanism $\pi^0$. Then $\pi$ weakly improves expected social welfare over $\pi^0$:
\[
\bbE_{\bmv \sim \calF}\big[\SW(h(\tbmp^*, \tbms^*; \bmv); \bmv)\big] \;\ge\; \bbE_{\bmv \sim \calF}\big[\SW(h(\tbmp^{0,*}, \tbms^{0,*}; \bmv); \bmv)\big].
\]
\end{corollary}

The proof, deferred to \cref{sec:omitted-proofs}, rests on a simple identity: under budget balance the monetary transfers cancel, so welfare equals the total realized valuation, $\SW(h;\bmv) = \sum_{i} y_i v_i$. By \cref{thm:global} each indicator $y_i$ is weakly larger under $\pi$ than under the baseline $\pi^0$, and positivity ($v_i > 0$) turns this into a weakly larger welfare.

\section{Experiments}
\label{sec:exp}

We complement the theory with exact computational experiments addressing four questions: (a) how much profit reallocation improves over the baseline on trees; (b) whether the gains survive dropping \cref{asp:homo,asp:posi}; (c) how they scale with the tree depth $d$ and branching factor $m$ in regular trees; and (d) whether they persist on an irregular tree.

\subsection{Setup}

We consider two families of tree structures. The first is \emph{complete $m$-ary trees} of depth $d$: the root originates the data and every non-leaf player has $m$ children, giving levels $0,\ldots,d$. By the level-symmetry of the tree under a Markovian valuation process, all players on the same level share the same equilibrium strategy, so it suffices to compute $d+1$ representative strategies. The second is the \emph{irregular} tree of \cref{fig:tree-notation}: a $7$-player instance with root $r$ and children $u,l$, where $l$ is a leaf, $u$ has children $uu$ (a leaf) and $ul$, and $ul$ has two leaf children $ulu,ull$---so branching and depth vary across nodes (depths range from $1$ to $3$). Lacking level-symmetry, it is solved without any symmetry reduction. All equilibria are computed exactly with \cref{alg:exact}.

\paragraph{Valuation models.}
Each child's valuation follows a local random walk on the grid: from a parent $i$ with value $v_i = w_k$, each child $j$ generates $k'$ independently in one of $k-2,\dots,k+2$ with equal probability (clipped into $[0,2d]$ for M2) and set $v_j = w_{k'}$. The four models share this transition and differ only in the values of $\{w_k\}$s. Specifically:
\begin{itemize}[left=0em]\itemsep0.1em
\item \textbf{M1} (homogeneous, positive): $w_k = e^{(k-2d)/2}$, $k=0,\ldots,4d$, root $v_r = w_{2d} = 1$; satisfies both \cref{asp:homo,asp:posi}.
\item \textbf{M2} (non-homogeneous, positive): $w_k = k/2d$, $k=0,\ldots,2d$, root $v_r = w_0 = 0$; violates homogeneity.
\item \textbf{M3} (homogeneous, non-positive): $w_k = (-1)^k e^{(k-2d)/2}$, $k=0,\ldots,4d$, root $v_r = w_{2d} = 1$; violates positivity.
\item \textbf{M4} (non-homogeneous, non-positive): $w_k = (k-2d)/2d$, $k=0,\ldots,4d$, root $v_r = w_{2d} = 0$; violates both.
\end{itemize}

\paragraph{Mechanisms.}
We compare four budget-balanced PRMs. For a trade $m$ with ancestors $\Path(m)$, writing $\ell_m \coloneqq |\Path(m)|$ and, for $k \in \Path(m)$, $s(m,k) \coloneqq \mathrm{depth}(\Parent(m)) - \mathrm{depth}(k) \in \{0,\ldots,\ell_m-1\}$ for the number of steps from the direct seller $\Parent(m)$ up to $k$:
\begin{itemize}[left=0em]\itemsep0.1em
\item \textbf{RAW} (baseline, \cref{def:baseline}): the direct seller keeps everything, $\pi(m,\Parent(m)) = 1$ and $\pi(m,k) = 0$ for $k \neq \Parent(m)$.
\item \textbf{TBDS}: a fixed-ratio rule with $\alpha = 0.5$ that shares a trade's payment geometrically up the ancestors---$\pi(m,k) = (1-\alpha)\,\alpha^{\,s(m,k)}$ for $k \neq r$, with the root absorbing the remainder $\pi(m,r) = \alpha^{\,\ell_m-1}$; this mechanism is originally proposed by Xin et al.~\cite{xin2025tbds} and is a special case of our PRMs.
\item \textbf{AVE}: each trade's payment is split uniformly among its ancestors, $\pi(m,k) = 1/\ell_m$ for every $k \in \Path(m)$.
\item \textbf{OPT}: a budget-balanced PRM optimized by a zeroth-order optimizer maximizing equilibrium welfare, run separately for each model.
\end{itemize}

Note that RAW, TBDS, and AVE are model-independent PRMs, while OPT relies on the underlying distribution $\calF$.

\paragraph{Metrics.}
We report two equilibrium quantities: \emph{social welfare} $\SW=\sum_{i\in\calI}u_i(\tbmp^*,\tbmx^*)$ and \emph{transaction volume} $\mathrm{TV}=\sum_{i\in\calI\setminus\{r\}}\Pr[y_i=1]$, the expected number of completed trades. Both are evaluated \emph{exactly} by the forward dynamic program, so every reported value is exact.

\subsection{Results}

\begin{table}[t]
\centering
\small
\caption{Equilibrium transaction volume (TV) and social welfare (SW) at $(d{=}5,\,m{=}2)$, computed exactly for RAW, TBDS~\citep{xin2025tbds}, AVE and OPT.}
\label{tab:exp}
\begin{tabular*}{\textwidth}{@{\extracolsep{\fill}} l cccc cccc @{}}
\toprule
& \multicolumn{4}{c}{Transaction Volume} & \multicolumn{4}{c}{Social Welfare} \\
\cmidrule(lr){2-5}\cmidrule(lr){6-9}
Mechanism & M1 & M2 & M3 & M4 & M1 & M2 & M3 & M4 \\
\midrule
RAW  & $2.69$  & $21.17$ & $0.66$  & $5.74$  & $36.70$  & $5.56$ & $7.47$  & $2.37$ \\
TBDS & $8.10$  & $26.03$ & $3.68$  & $9.94$  & $75.99$  & $6.76$ & $14.03$ & $3.61$ \\
AVE  & $12.83$ & $37.27$ & $4.48$  & $12.91$ & $92.81$  & $8.23$ & $19.72$ & $4.33$ \\
OPT  & $42.80$ & $46.33$ & $16.81$ & $20.15$ & $140.00$ & $9.33$ & $39.14$ & $5.34$ \\
\bottomrule
\end{tabular*}
\end{table}

\paragraph{Welfare and trade improvements (a, b).}
\cref{tab:exp} reports both metrics for a representative configuration ($d=5$, $m=2$). Every PRM improves over RAW across all four models, including those violating \cref{asp:homo,asp:posi}. Averaged over the four models, the model-independent AVE and TBDS already yield large gains over RAW---$289\%/112\%$ and $189\%/67\%$ in transaction volume/social welfare.
The results suggest that PRMs remain effective across diverse scenarios, even when \cref{asp:homo,asp:posi} no longer hold.

\begin{figure}[t]
\centering
\includegraphics[width=\textwidth]{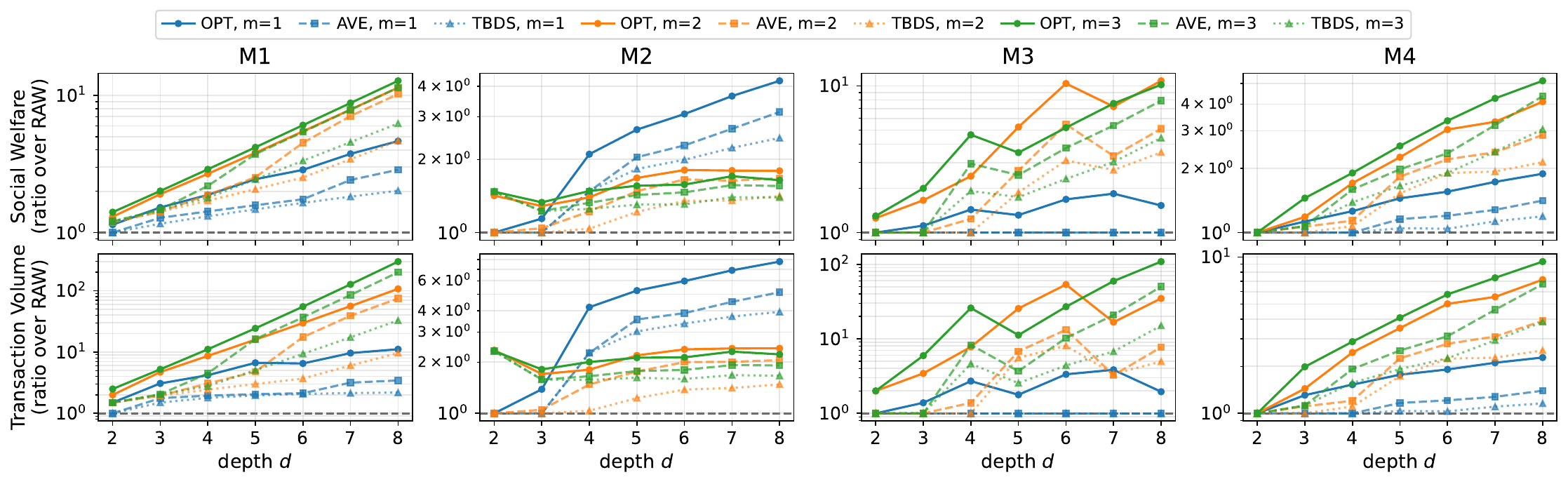}
\caption{Equilibrium social welfare (top) and transaction volume (bottom) of the TBDS (dotted, $\triangle$), AVE (dashed, $\square$), and OPT (solid, $\circ$) mechanisms, expressed as a ratio over RAW, as the tree depth $d$ varies from $2$ to $8$. Each column represents one valuation model; colors denote the branching factor $m \in \{1,2,3\}$. Ratios above the dotted line ($=1$) indicate that profit reallocation improves upon the baseline. The results show that the advantage of PRMs generally grows with depth and branching.}
\label{fig:exp-vary}
\end{figure}

\paragraph{Scaling with depth and branching (c).}
Varying $2 \le d\le 8$ and $m\in\{1,2,3\}$ with results shown in \cref{fig:exp-vary}, the advantage of PRMs over RAW generally grows with both depth $d$ and branching $m$---deeper, bushier trees create more downstream profits to reallocate---and is robust to violations of positivity and homogeneity.

\paragraph{Irregular trees (d).}
We report the results of the \emph{irregular} tree structure as in \cref{fig:tree-notation}, shown in \cref{tab:exp-irregular}, \cref{sec:app-experiments}. The picture is unchanged---PRMs beat RAW for every model, indicating that the benefit of PRMs is not specific to symmetric trees.

\section{Conclusion}
\label{sec:conclusion}

We studied profit reallocation on tree-structured data markets that capture one-to-many resale. We introduced a tree-structured trading game and a general, budget-feasible profit reallocation mechanism, gave a polynomial-time exact algorithm and an FPTAS for the induced equilibria, and proved that reallocation weakly expands equilibrium trade and improves social welfare over the no-reallocation baseline. Open directions include extending from trees to directed acyclic graphs, endogenizing the trading topology, treating richer trading protocols beyond posting-price, studying collusion- or sybil-robust solution concepts, characterizing when a PRM is optimal and designing algorithms for computing approximately optimal PRMs, and objectives beyond welfare such as revenue or fairness.

\subsubsection*{Acknowledgments.}
This work was supported by China Unicom Software Research Institute and the National Natural Science Foundation of China (Grant No. 62572010).

\bibliographystyle{splncs04}
\bibliography{_reference}

@article{ma2026incentivizing,
  title={Incentivizing Data Trading via Profit Reallocation},
  author={Ma, Yunxuan and Xin, Wu and Li, Jichen and Chen, Hongyin and Dong, Xiaoqi and Zheng, Yusen and Cheng, Yukun and Deng, Xiaotie},
  journal={arXiv preprint arXiv:2606.31202},
  year={2026}
}

@article{myerson1981optimal,
  title={Optimal auction design},
  author={Myerson, Roger B},
  journal={Mathematics of Operations Research},
  volume={6},
  number={1},
  pages={58--73},
  year={1981},
  publisher={INFORMS}
}

@article{analytics2016age,
  title={The age of analytics: competing in a data-driven world},
  author={Analytics, McKinsey},
  journal={McKinsey Global Institute Research},
  year={2016}
}

@INPROCEEDINGS{cao2017game,
  author={Cao, Zhi and Zhang, Honggang and Liu, Benyuan and Sheng, Bo},
  booktitle={2018 IEEE 37th International Performance Computing and Communications Conference}, 
  title={A Game-theoretic Framework for Revenue Sharing in Edge-Cloud Computing System}, 
  year={2018},
  volume={},
  number={},
  pages={1-8},
  doi={10.1109/PCCC.2018.8710866}
}

@article{richter2019data,
  title={The data sharing economy: on the emergence of new intermediaries},
  author={Richter, Heiko and Slowinski, Peter R},
  journal={IIC-International Review of Intellectual Property and Competition Law},
  volume={50},
  number={1},
  pages={4--29},
  year={2019},
  publisher={Springer}
}

@inproceedings{ghosh2011selling,
  title={Selling privacy at auction},
  author={Ghosh, Arpita and Roth, Aaron},
  booktitle={Proceedings of the 12th ACM Conference on Electronic Commerce},
  pages={199--208},
  year={2011}
}

@inproceedings{nissim2014redrawing,
  title={Redrawing the boundaries on purchasing data from privacy-sensitive individuals},
  author={Nissim, Kobbi and Vadhan, Salil and Xiao, David},
  booktitle={Proceedings of the 5th conference on Innovations in theoretical computer science},
  pages={411--422},
  year={2014}
}

@inproceedings{zhang2020selling,
  title={Selling data at an auction under privacy constraints},
  author={Zhang, Mengxiao and Beltran, Fernando and Liu, Jiamou},
  booktitle={Conference on Uncertainty in Artificial Intelligence},
  pages={669--678},
  year={2020},
  organization={PMLR}
}

@inproceedings{yu2020fairness,
  title={A fairness-aware incentive scheme for federated learning},
  author={Yu, Han and Liu, Zelei and Liu, Yang and Chen, Tianjian and Cong, Mingshu and Weng, Xi and Niyato, Dusit and Yang, Qiang},
  booktitle={Proceedings of the AAAI/ACM Conference on AI, Ethics, and Society},
  pages={393--399},
  year={2020}
}

@article{ali2024reselling,
  title={Reselling information},
  author={Ali, S Nageeb and Chen-Zion, Ayal and Lillethun, Erik},
  journal={Games and Economic Behavior},
  volume={148},
  pages={23--43},
  year={2024},
  publisher={Elsevier}
}

@article{shiller2013digital,
  title={Digital distribution and the prohibition of resale markets for information goods},
  author={Shiller, Benjamin Reed},
  journal={Quantitative Marketing and Economics},
  volume={11},
  number={4},
  pages={403--435},
  year={2013},
  publisher={Springer}
}

@article{condorelli2017bilateral,
  title={Bilateral trading in networks},
  author={Condorelli, Daniele and Galeotti, Andrea and Renou, Ludovic},
  journal={The Review of Economic Studies},
  volume={84},
  number={1},
  pages={82--105},
  year={2017},
  publisher={Oxford University Press}
}

@article{manea2018intermediation,
  title={Intermediation and resale in networks},
  author={Manea, Mihai},
  journal={Journal of Political Economy},
  volume={126},
  number={3},
  pages={1250--1301},
  year={2018},
  publisher={University of Chicago Press Chicago, IL}
}

@inproceedings{emek2011mechanisms,
  title={Mechanisms for multi-level marketing},
  author={Emek, Yuval and Karidi, Ron and Tennenholtz, Moshe and Zohar, Aviv},
  booktitle={Proceedings of the 12th ACM conference on Electronic commerce},
  pages={209--218},
  year={2011}
}

@inproceedings{li2017mechanism,
  title={Mechanism design in social networks},
  author={Li, Bin and Hao, Dong and Zhao, Dengji and Zhou, Tao},
  booktitle={Proceedings of the AAAI Conference on Artificial Intelligence},
  volume={31},
  number={1},
  year={2017}
}

@article{li2022diffusion,
  title={Diffusion auction design},
  author={Li, Bin and Hao, Dong and Gao, Hui and Zhao, Dengji},
  journal={Artificial Intelligence},
  volume={303},
  pages={103631},
  year={2022},
  publisher={Elsevier}
}

@inproceedings{zhang2020redistribution,
  title={Redistribution Mechanism on Networks},
  author={Zhang, Wen and Zhao, Dengji and Chen, Hanyu},
  booktitle={Proceedings of the 19th International Conference on Autonomous Agents and MultiAgent Systems},
  pages={1620--1628},
  year={2020}
}

@article{liu2024data,
  title={Data sharing and exchanging with incentive and optimization: a survey},
  author={Liu, Liyuan and Han, Meng},
  journal={Discover Data},
  volume={2},
  number={1},
  pages={2},
  year={2024},
  publisher={Springer}
}

@article{zhang2025incentive,
  title={An Incentive Mechanism for Privacy Preserved Data Trading with Verifiable Data Disturbance},
  author={Zhang, Man and Li, Xinghua and Luo, Bin and Ren, Yanbing and Miao, Yinbin and Liu, Ximeng and Deng, Robert H},
  journal={IEEE Transactions on Dependable and Secure Computing},
  year={2025},
  publisher={IEEE}
}

@article{han2024dynamic,
  title={Dynamic incentive design for federated learning based on consortium blockchain using a stackelberg game},
  author={Han, Baofu and Li, Bing and Wolter, Katinka and Jurdak, Raja and Zhang, Hao and Hu, Yuanyuan and Li, Yi},
  journal={IEEE Access},
  year={2024},
  publisher={IEEE}
}

@inproceedings{biswas2021incentive,
  title={An incentive mechanism for trading personal data in data markets},
  author={Biswas, Sayan and Jung, Kangsoo and Palamidessi, Catuscia},
  booktitle={International Colloquium on Theoretical Aspects of Computing},
  pages={197--213},
  year={2021},
  organization={Springer}
}

@inproceedings{xin2025tbds,
  author       = {Wu Xin and
                  Hongyin Chen and
                  Xiaoqi Dong and
                  Jichen Li and
                  Xiaotie Deng and
                  Zhonghai Wu and
                  Bin Xiao},
  title        = {{TBDS:} Transaction-Based Data Sharing},
  booktitle    = {Frontiers of Algorithmics - 19th International Joint Conference},
  volume       = {15828},
  pages        = {222--237},
  publisher    = {Springer},
  year         = {2025},
  doi          = {10.1007/978-981-96-8312-3\_17}
}

@inproceedings{chen2017model,
  title={Model-based pricing: Do not pay for more than what you learn!},
  author={Chen, Lingjiao and Koutris, Paraschos and Kumar, Arun},
  booktitle={Proceedings of the 1st Workshop on Data Management for End-to-end Machine Learning},
  pages={1--4},
  year={2017}
}

@inproceedings{chen2019towards,
  title={Towards model-based pricing for machine learning in a data marketplace},
  author={Chen, Lingjiao and Koutris, Paraschos and Kumar, Arun},
  booktitle={Proceedings of the 2019 international conference on management of data},
  pages={1535--1552},
  year={2019}
}

@article{azcoitia2022survey,
  title={A survey of data marketplaces and their business models},
  author={Azcoitia, Santiago Andr{\'e}s and Laoutaris, Nikolaos},
  journal={ACM SIGMOD Record},
  volume={51},
  number={3},
  pages={18--29},
  year={2022},
  publisher={ACM New York, NY, USA}
}

@inproceedings{li2023martfl,
  title={martfl: Enabling utility-driven data marketplace with a robust and verifiable federated learning architecture},
  author={Li, Qi and Liu, Zhuotao and Li, Qi and Xu, Ke},
  booktitle={Proceedings of the 2023 ACM SIGSAC Conference on Computer and Communications Security},
  pages={1496--1510},
  year={2023}
}

@article{mi2025multi,
  title={Multi-Party Data Pricing for Complex Data Trading Markets: A Rubinstein Bargaining Approach},
  author={Mi, Bing and Han, Zhengwang and Chen, Kongyang},
  journal={arXiv preprint arXiv:2502.16363},
  year={2025}
}

@inproceedings{qin2023research,
  title={Research on pricing mechanism of data products based on game theory},
  author={Qin, Siyao and Si, Yaqing and He, Bin and Yan, Hongyan and Li, Wen and Shi, Xuemei},
  booktitle={International Conference on Electronic Information Engineering and Data Processing (EIEDP 2023)},
  volume={12700},
  pages={158--165},
  year={2023},
  organization={SPIE}
}

@article{bauer2024designing,
  title={Designing a blockchain-based data market and pricing data to optimize data trading and welfare},
  author={Bauer-H{\"a}nsel, Ingrid and Liu, Qianyu and Tessone, Claudio J and Schwabe, Gerhard},
  journal={International Journal of Electronic Commerce},
  volume={28},
  number={1},
  pages={3--30},
  year={2024},
  publisher={Taylor \& Francis}
}

@inproceedings{sandholm2006sequences,
  title={Sequences of take-it-or-leave-it offers: Near-optimal auctions without full valuation revelation},
  author={Sandholm, Tuomas and Gilpin, Andrew},
  booktitle={Proceedings of the fifth international joint conference on Autonomous agents and multiagent systems},
  pages={1127--1134},
  year={2006}
}

@misc{OECD_EASD_2019,
  title        = {Enhancing Access to and Sharing of Data: Reconciling Risks and Benefits for Data Re-use across Societies},
  author       = {{OECD}},
  year         = {2019},
  publisher    = {OECD Publishing},
  doi          = {10.1787/276aaca8-en},
  url          = {https://www.oecd.org/content/dam/oecd/en/publications/reports/2019/11/enhancing-access-to-and-sharing-of-data_070835df/276aaca8-en.pdf}
}

@book{AGT:nisan_algorithmic_2007,
	address = {Cambridge},
	title = {Algorithmic {Game} {Theory}},
	isbn = {978-0-521-87282-9},
	url = {https://www.cambridge.org/core/books/algorithmic-game-theory/0092C07CA8B724E1B1BE2238DDD66B38},
	urldate = {2025-06-01},
	publisher = {Cambridge University Press},
	editor = {Nisan, Noam and Roughgarden, Tim and Tardos, Eva and Vazirani, Vijay V.},
	year = {2007},
	doi = {10.1017/CBO9780511800481},
}

@book{MWG,
  title={Microeconomic theory},
  author={Mas-Colell, Andreu and Whinston, Michael Dennis and Green, Jerry R and others},
  volume={1},
  year={1995},
  publisher={Oxford university press New York}
}

@book{fudenberg1991game,
  title={Game Theory},
  author={Fudenberg, Drew and Tirole, Jean},
  year={1991},
  publisher={MIT Press},
  address={Cambridge, MA}
}

@article{agarwal2024towards,
  title={Towards data auctions with externalities},
  author={Agarwal, Anish and Dahleh, Munther and Horel, Thibaut and Rui, Maryann},
  journal={Games and Economic Behavior},
  volume={148},
  pages={323--356},
  year={2024},
  publisher={Elsevier}
}

@techreport{OECD2021ValueOfData,
  title  = {Measuring the Economic Value of Data},
  author = {{OECD}},
  year   = {2021},
  institution = {OECD}
}

@article{JonesTonetti2020NonrivalryData,
  title   = {Nonrivalry and the Economics of Data},
  author  = {Jones, Charles I. and Tonetti, Christopher},
  journal = {American Economic Review},
  year    = {2020},
  volume  = {110},
  number  = {9},
  pages   = {2819--2858},
  doi     = {10.1257/aer.20191330}
}

@article{DSGE,
  title={Dynamic Stochastic General Equilibrium Models},
  author={Gil, Hamilton Galindo and Bravo, Alexis Montecinos and Sosa, Marco Antonio Ortiz},
  journal={Springer Texts in Business and Economics},
  year={2024},
  publisher={Springer}
}

@misc{SnowflakeDocs_MarketplaceAbout,
  title        = {About Snowflake Marketplace},
  author       = {{Snowflake Inc.}},
  howpublished = {Snowflake Documentation},
  url          = {https://docs.snowflake.com/en/collaboration/collaboration-marketplace-about}
}

@misc{AWSDocs_DataExchangeLanding,
  title        = {AWS Data Exchange Documentation},
  author       = {{Amazon Web Services, Inc.}},
  howpublished = {AWS Documentation},
  url          = {https://docs.aws.amazon.com/data-exchange/}
}

@misc{Dawex_Web_DataMarketplaceModel,
  title        = {Data marketplace: {Dawex} Data Exchange Technology},
  author       = {{Dawex Systems}},
  howpublished = {Dawex official website},
  url          = {https://www.dawex.com/en/data-exchange-platform/data-marketplace/}
}

@techreport{Koutroumpis2017UnfulfilledPotential,
  author      = {Koutroumpis, Pantelis and Leiponen, Aija and Thomas, Llewellyn D. W.},
  title       = {The (Unfulfilled) Potential of Data Marketplaces},
  institution = {ETLA -- The Research Institute of the Finnish Economy},
  type        = {ETLA Working Papers},
  number      = {53},
  year        = {2017},
  url         = {https://www.etla.fi/wp-content/uploads/etla-working-papers-53.pdf}
}

@misc{EuropeanParliament2023DataActPressRelease,
  author       = {{European Parliament}},
  title        = {Data {Act}: MEPs back new rules for fair access to and use of industrial data},
  year         = {2023},
  month        = mar,
  day          = {14},
  url          = {https://www.europarl.europa.eu/news/en/press-room/20230310IPR77226/data-act-meps-back-new-rules-for-fair-access-to-and-use-of-industrial-data}
}

\newpage
\appendix
\section*{Appendices}
\section{Omitted Proofs and Pseudocode}
\label{sec:omitted-proofs}

\subsection{Omitted Pseudocode}
We defer to here the two subroutines referenced in \cref{sec:algo}: the forward dynamic program \textsc{ForwardDP} (\cref{alg:forward-dp}), which evaluates the gain function $G^{(j)}_i$ called by \cref{alg:exact}, and the FPTAS (\cref{alg:approx}) of \cref{thm:algo-approx}.

\begin{algorithm}[h]
\caption{ForwardDP: Compute $G^{(j)}_i(w_k, p_j)$}
\label{alg:forward-dp}
\KwIn{Tree $\calT$, PRM $\pi$, transition $Q$, subtree root $j$, seller $i$, valuation index $k$, price $p_j$}
\KwOut{$G^{(j)}_i(w_k, p_j)$}
\For{$\ell = 1, \ldots, S$}{
  $\Pr[Y(j, \ell)] \gets q_{j,k,\ell} \cdot \one\{s^*(j, \ell) \ge p_j\}$\;
}
$\text{gain} \gets \pi(j, i) \cdot p_j \cdot \sum_{\ell \in [S]} \Pr[Y(j, \ell)]$\;
\For{each $m \in \Subtree(j)$ \textup{(top-down BFS order)}}{
  \For{each $m' \in \Children(m)$}{
    \For{$\ell' = 1, \ldots, S$}{
      $\Pr[Y(m', \ell')] \gets \sum_{\ell \in [S]} \Pr[Y(m, \ell)] \cdot q_{m',\ell,\ell'} \cdot \one\{s^*(m', \ell') \ge p^*(m', \ell)\}$\;
    }
    $\text{gain} \gets \text{gain} + \sum_{\ell \in [S]} \pi(m', i) \cdot p^*(m', \ell) \cdot \Pr[Y(m, \ell)] \cdot \sum_{\ell'} q_{m',\ell,\ell'} \cdot \one\{s^*(m', \ell') \ge p^*(m', \ell)\}$\;
  }
}
\textbf{return} gain\;
\end{algorithm}

\begin{algorithm}[t]
\caption{FPTAS for $\varepsilon$-Approximate Equilibrium}
\label{alg:approx}
\KwIn{$\varepsilon > 0$, tree $\calT$, PRM $\pi$, conditional CDFs $\{F_i\}$, support $[L_0, H_0]$, Lipschitz constant $K_0$}
\KwOut{$\varepsilon$-approximate equilibrium $(\tbmp^*, \tbms^*)$}
$R \gets \max(|L_0|, |H_0|)$\;
$\varepsilon_0 \gets \frac{\min(\varepsilon, 1)}{8(1 + dK_0)(1 + R)}$\;
$S \gets \lfloor (H_0 - L_0)/\varepsilon_0 \rfloor + 1$\;
$W \gets \{L_0 + \ell \cdot \varepsilon_0\}_{\ell = 0}^{S-1}$\;
\For{each $i \in \calI \setminus \{r\}$, $k, \ell \in \{0, \ldots, S-1\}$}{
  $q_{i,k,\ell} \gets F_i(w_{\ell+1} \mid w_k) - F_i(w_\ell \mid w_k)$\;
}
$\{s^*(i,k)\}, \{p^*(j,k)\} \gets \text{\cref{alg:exact}}(\calT, \pi, Q, W)$\;
\For{each $i \in \calI \setminus \{r\}$}{
  $\ts^*_i(v_i) \gets s^*(i, \lfloor (v_i - L_0)/\varepsilon_0 \rfloor)$\;
  $\tp^*_j(v_i) \gets p^*(j, \lfloor (v_i - L_0)/\varepsilon_0 \rfloor)$ for each $j \in \Children(i)$\;
}
\textbf{return} $(\tbmp^*, \tbms^*)$\;
\end{algorithm}

\subsection{Proof of \texorpdfstring{\cref{fact:simplify}}{Fact~\ref{fact:simplify}}: Posterior Simplification}
\begin{proof}
Fix the strategy profile $(\tbmp, \tbmx)$. By the game dynamics (\cref{def:game-evolve}), every price and decision recorded in $h_i$ (see~\eqref{eq:hi}) is computed recursively from the valuations of the strict ancestors of $i$; hence there is a deterministic map $\phi_i$ with $h_i = \phi_i(\bmv_{\Path(i)})$, where $\bmv_{\Path(i)} = \{v_\ell\}_{\ell \in \Path(i)}$. It therefore suffices to prove the conditional independence $\bmv_{\Desc(i)} \perp \bmv_{\Path(i)} \mid v_i$.

Write $A = \Path(i)$, $D = \Desc(i)$, and $O = \calI \setminus (A \cup \{i\} \cup D)$ for the remaining nodes. By \cref{asp:markov}, the joint density factorizes as
\[
f(\bmv) = \underbrace{f_r(v_r)\!\!\prod_{\ell \in A \setminus \{r\}}\!\! f_\ell(v_\ell \mid v_{\Parent(\ell)})}_{g(\bmv_A)} \cdot\; f_i(v_i \mid v_{\Parent(i)}) \cdot \underbrace{\prod_{\ell \in D} f_\ell(v_\ell \mid v_{\Parent(\ell)})}_{\eta(v_i, \bmv_D)} \cdot \prod_{\ell \in O} f_\ell(v_\ell \mid v_{\Parent(\ell)}),
\]
where we used that every $\ell \in D$ has $\Parent(\ell) \in \{i\} \cup D$, so the factor $\eta$ depends only on $(v_i, \bmv_D)$. Each node in $O$ lies in a subtree hanging off some node of $A$, so integrating out $\bmv_O$ over a complete collection of conditional densities gives $\int \prod_{\ell \in O} f_\ell(v_\ell \mid v_{\Parent(\ell)}) \, d\bmv_O = 1$. Thus
\[
f(\bmv_A, v_i, \bmv_D) = g(\bmv_A)\, f_i(v_i \mid v_{\Parent(i)})\, \eta(v_i, \bmv_D),
\]
and dividing by $\int f(\bmv_A, v_i, \bmv_D)\, d\bmv_D = g(\bmv_A)\, f_i(v_i \mid v_{\Parent(i)}) \int \eta(v_i, \bmv_D)\, d\bmv_D$ yields
\[
f(\bmv_D \mid v_i, \bmv_A) = \frac{\eta(v_i, \bmv_D)}{\int \eta(v_i, \bmv_D)\, d\bmv_D},
\]
which is independent of $\bmv_A$. Hence $\calF(\bmv_D \mid v_i, \bmv_A) = \calF(\bmv_D \mid v_i)$, and since $h_i = \phi_i(\bmv_A)$, we conclude $\calF(\bmv_{\Desc(i)} \mid v_i, h_i) = \calF(\bmv_{\Desc(i)} \mid v_i)$.
\qed
\end{proof}

\subsection{Proof of \texorpdfstring{\cref{lem:simplify-sell}}{Lemma~\ref{lem:simplify-sell}}: Equilibrium Simplification}
\lemSimplifySell*
\begin{proof}
We construct a candidate profile $(\tbmp^\dagger, \tbmx^\dagger)$ of the claimed simplified form together with a Bayes-consistent belief system, and verify the two requirements of \cref{def:pbe}. We proceed by backward induction on the tree (leaves to root), showing that at \emph{every} information set $(h_i, v_i)$ the prescribed action maximizes player $i$'s expected utility $\bbU_i$~\eqref{eq:cond-utility} while all other players follow the profile---that is, the profile is sequentially rational. Throughout, the payoff-relevant belief is the Bayes-consistent posterior $\calF(\cdot \mid v_i)$ over $\bmv_{\Desc(i)}$ at every information set (\cref{fact:simplify}), so part~(b) of \cref{def:pbe} holds automatically and we only argue part~(a).

\paragraph{Base case (leaf nodes).}
For any leaf $i \in \calL$, player $i$ has no children, so $\Desc(i) = \varnothing$. Player $i$'s utility from \cref{eq:utility-raw} reduces to:
\[
U_i(h; v_i) \;=\; y_i \cdot (v_i - p_i).
\]
Expanding $y_i = y_{\Parent(i)} \cdot x_i$ via \cref{eq:yi}, the expected utility conditioned on history $h_i$ and action $x_i$ is:
\[
\bbE[U_i \mid h_i, v_i, x_i] \;=\; x_i \cdot y_{\Parent(i)} \cdot (v_i - p_i),
\]
where $y_{\Parent(i)} \in \{0,1\}$ and $p_i \ge 0$ are both determined by $h_i$.

\emph{Case $y_{\Parent(i)} = 0$:} The utility is $0$ regardless of $x_i$, so any $x_i$ is optimal.

\emph{Case $y_{\Parent(i)} = 1$:} The utility equals $x_i(v_i - p_i)$. Since $x_i \in \{0,1\}$, the unique optimum is $x_i = \one\{v_i \ge p_i\}$.

In both cases, the optimal $x_i$ depends only on $(v_i, p_i)$, not on the full history $h_i$. Hence $\tx_i(h_i, v_i) = \tx^\dagger_i(v_i, p_i)$ for some function $\tx^\dagger_i$.

\paragraph{Inductive step.}
Fix a non-leaf node $i$. As inductive hypothesis, assume that for all $m \in \Desc(i)$ the strategies have been simplified---$\tp^\dagger_m$ depends only on $v_m$ and $\tx^\dagger_m$ only on $(v_m, p_m)$---and are moreover \emph{measurable and bounded}; equivalently, each threshold $\ts^\dagger_m$ (\cref{cor:simplify-buy}) is measurable with $\sup_v |\ts^\dagger_m(v)| < \infty$. This holds at the leaves, where $\ts^\dagger_m(v) = v$ is bounded on the compact $\calV_m$. We show below that node $i$ inherits the same properties, closing the induction.

By \cref{fact:simplify} and the Markov property (\cref{asp:markov}), the posterior distribution satisfies $\calF(\bmv_{\Desc(i)} \mid v_i, h_i) = \calF(\bmv_{\Desc(i)} \mid v_i)$. Under the inductive hypothesis, the strategies of all descendants are functions of their own valuations and local prices only, so the expected utility from \cref{eq:utility-raw} expands as:
\begin{align}
&\bbE[U_i \mid h_i, v_i, x_i, \{p_j\}_{j \in \Children(i)}] \nonumber\\
&= x_i \cdot y_{\Parent(i)} \cdot (v_i - p_i) + \bbE_{\bmv_{\Desc(i)} \sim \calF(\cdot|v_i)}\bigg[\sum_{m \in \Desc(i)} \pi(m,i) \cdot p_m \cdot y_m\bigg]. \label{eq:prf:simplify:expand}
\end{align}
For any $m \in \Desc(i)$, the trade indicator decomposes as $y_m = y_i \cdot y_{m|i} = x_i \cdot y_{\Parent(i)} \cdot y_{m|i}$, where $y_{m|i}$ (\cref{eq:y-conditional}) depends only on valuations and strategies within $\Desc(i)$. Substituting into~\eqref{eq:prf:simplify:expand} and factoring out $x_i \cdot y_{\Parent(i)}$:
\begin{equation}
\label{eq:prf:simplify:factor}
\bbE[U_i \mid h_i, v_i, x_i, \{p_j\}] = x_i \cdot y_{\Parent(i)} \cdot \bigg[(v_i - p_i) + \sum_{j \in \Children(i)} G^{(j)}_i(v_i, p_j)\bigg],
\end{equation}
where $G^{(j)}_i(v_i, p_j)$ is defined in \cref{eq:gain-child}. Each $G^{(j)}_i(v_i, p_j)$ depends only on $v_i$ and $p_j$, since the strategies and valuations within $\Subtree(j)$ are independent of $h_i$ by the inductive hypothesis.

\emph{Pricing.} The factor $x_i \cdot y_{\Parent(i)} \ge 0$ does not depend on $\{p_j\}$, so maximizing $\bbE[U_i \mid \cdot]$ over $\{p_j\}$ reduces to maximizing each $G^{(j)}_i(v_i, p_j)$ independently over $p_j \ge 0$. As $G^{(j)}_i(v_i,\cdot)$ depends only on $v_i$ and the strategies within $\Subtree(j)$, the optimizer is a function of $v_i$ alone; we now show it is attained and well-behaved, which both defines $\tp^\dagger_j$ and propagates the inductive hypothesis.

\emph{Attainment and regularity.} Write, from \cref{eq:gain-child},
\[
G^{(j)}_i(v_i, p) = \pi(j,i)\,p\,\Pr[\ts^\dagger_j(v_j) \ge p \mid v_i] \;+\; \bbE_{\bmv_{\Subtree(j)} \sim \calF(\cdot \mid v_i)}\!\big[\one\{\ts^\dagger_j(v_j) \ge p\}\, D\big],
\]
where $D \coloneqq \sum_{m \in \Desc(j)} \pi(m,i)\, \tp^\dagger_m(v_{\Parent(m)})\, y_{m\mid j} \ge 0$ is independent of $p$ and, by the inductive hypothesis (bounded prices) with $\pi(\cdot,\cdot) \le 1$, bounded. Since $\calV_j$ is compact and $\ts^\dagger_j$ bounded, $\bar s_j \coloneqq \sup_{v}\ts^\dagger_j(v) < \infty$, and any $p > \bar s_j$ gives $G^{(j)}_i(v_i,p) = 0$; the search thus reduces to the compact interval $[0, \bar s_j]$.

For each fixed $v_j$, the acceptance indicator $p \mapsto \one\{\ts^\dagger_j(v_j) \ge p\} = \one_{[0,\, \ts^\dagger_j(v_j)]}(p)$ is upper semicontinuous in $p$ (its superlevel sets are closed), thanks to the accept-when-indifferent convention $\{\ts^\dagger_j \ge p\}$. These indicators are uniformly bounded by $1$ and $D$ is bounded, so the reverse Fatou lemma yields, for every $p \in [0, \bar s_j]$, $\limsup_{p' \to p} G^{(j)}_i(v_i, p') \le G^{(j)}_i(v_i, p)$; that is, $G^{(j)}_i(v_i, \cdot)$ is upper semicontinuous on $[0,\bar s_j]$ (the continuous prefactor $\pi(j,i)p$ preserves this). An upper semicontinuous function on a compact set attains its maximum, so $\argmax_{p \ge 0} G^{(j)}_i(v_i, \cdot) \neq \varnothing$; we select its smallest element as $\tp^\dagger_j(v_i)$. Because $G^{(j)}_i$ is jointly measurable in $(v_i,p)$ and upper semicontinuous in $p$, the measurable maximum theorem makes both $\tp^\dagger_j(v_i)$ and $\max_p G^{(j)}_i(v_i,\cdot)$ measurable in $v_i$, and both are bounded ($\tp^\dagger_j \in [0,\bar s_j]$ and $0 \le \max_p G^{(j)}_i \le \bar s_j + \bbE[D]$). Hence $\tp_i(h_i, j, v_i) = \tp^\dagger_j(v_i)$ is a measurable, bounded function of $v_i$ alone.

Crucially, this argument needs only the measurability and boundedness of the $\Subtree(j)$ strategies---which the induction already carries---together with the accept-when-indifferent convention; it never requires their continuity or any joint regularity, so no circularity arises in the backward induction.

\emph{Buying.} After substituting the optimal prices $\{p^*_j(v_i)\}_{j \in \Children(i)}$, define:
\[
\Phi(v_i) \;\coloneqq\; (v_i - p_i) + \sum_{j \in \Children(i)} G^{(j)}_i\big(v_i, p^*_j(v_i)\big).
\]
Then $\bbE[U_i \mid \cdot] = x_i \cdot y_{\Parent(i)} \cdot \Phi(v_i)$. When $y_{\Parent(i)} = 0$, the utility is $0$ for all $x_i$. When $y_{\Parent(i)} = 1$, the optimal choice is $x_i = \one\{\Phi(v_i) \ge 0\}$, which depends only on $(v_i, p_i)$. Hence $\tx_i(h_i, v_i) = \tx^\dagger_i(v_i, p_i)$. The induced threshold $\ts^\dagger_i(v_i) = v_i + \sum_{j \in \Children(i)} \max_{p} G^{(j)}_i(v_i, p)$ (\cref{cor:simplify-buy}) is measurable and bounded on the compact $\calV_i$, since each $\max_p G^{(j)}_i$ is measurable and bounded; thus node $i$ satisfies the inductive hypothesis, closing the induction.

\paragraph{From the construction to a perfect Bayesian equilibrium.}
The backward induction defines, for every player $i$, the simplified strategies $\tp^\dagger_j$ ($j \in \Children(i)$) and $\tx^\dagger_i$. By construction, the prescribed actions maximize player $i$'s expected utility $\bbU_i\big((\tbmp^\dagger, \tbmx^\dagger), \calF(\cdot\mid v_i) \mid h_i, v_i\big)$~\eqref{eq:cond-utility} over player $i$'s own actions at \emph{every} information set $(h_i, v_i)$, while all other players follow $(\tbmp^\dagger, \tbmx^\dagger)$; in particular, this optimality holds for every offered price $p_i$ and every value of $y_{\Parent(i)}$, hence at on- and off-path information sets alike (the root and leaf cases follow from the same computation restricted to pricing or buying alone). Equivalently, for every deviation $(\tp_i, \tx_i)$,
\[
\bbU_i\big((\tbmp^\dagger, \tbmx^\dagger), \calF(\cdot\mid v_i) \mid h_i, v_i\big) \;\ge\; \bbU_i\big((\tp_i, \tbmp^\dagger_{-i}, \tx_i, \tbmx^\dagger_{-i}), \calF(\cdot\mid v_i) \mid h_i, v_i\big),
\]
which is exactly the sequential-rationality requirement~(a) of \cref{def:pbe}. Together with the Bayes-consistency of the beliefs $\calF(\cdot \mid v_i)$ (\cref{fact:simplify}), this shows that $\big((\tbmp^\dagger, \tbmx^\dagger), \mu\big)$ is a perfect Bayesian equilibrium of the simplified form~\eqref{eq:lem:simplify-sell}.
\qed
\end{proof}

\subsection{Proof of \texorpdfstring{\cref{cor:simplify-buy}}{Corollary~\ref{cor:simplify-buy}}: Threshold Strategy}
\corSimplifyBuy*
\begin{proof}
By \cref{lem:simplify-sell}, buyer $i$'s optimal action is $x_i = 1$ iff $v_i - p_i + \sum_{j \in \Children(i)} \max_{p_j} G^{(j)}_i(v_i, p_j) \ge 0$. Define:
\[
\ts^\dagger_i(v_i) \;\coloneqq\; v_i + \sum_{j \in \Children(i)} \max_{p_j \ge 0}\; G^{(j)}_i(v_i, p_j).
\]
Then $x_i = 1$ iff $\ts^\dagger_i(v_i) \ge p_i$, which is the threshold strategy in \cref{eq:threshold}. For leaf $i$: $\ts^\dagger_i(v_i) = v_i$.
\qed
\end{proof}

\subsection{Proof of \texorpdfstring{\cref{lem:subtree-decomp}}{Lemma~\ref{lem:subtree-decomp}}: Subtree Decomposition}
\lemSubtreeDecomp*
\begin{proof}
We verify the two requirements of \cref{def:pbe} for the simplified profile $(\tbmp^*, \tbms^*)$ paired with the Bayes-consistent beliefs $\calF(\cdot \mid v_i)$ (\cref{fact:simplify}); the latter discharges part~(b), so it suffices to establish sequential rationality, part~(a)---that conditions~\eqref{eq:pricing-decomp}--\eqref{eq:threshold-decomp} make the prescribed actions optimal at every information set. We argue for a non-leaf, non-root player $i$ and fix $v_i \in \calV_i$; the root and leaf cases are analogous, involving only pricing or only buying. Consider an arbitrary deviation $(\tp_i, \tx_i)$ for player $i$, with all other players following $(\tbmp^*_{-i}, \tbmx^*_{-i})$, at an information set with observed history $h_i$, which by \cref{fact:simplify} is a function of $\bmv_{\Path(i)}$ and determines the offered price $p_i$ and the indicator $y_{\Parent(i)} \in \{0,1\}$. We first decompose player $i$'s expected selling profit into per-child gains (Steps 1--3), then identify the optimal actions and conclude sequential rationality (Step 4).

\paragraph{Step 1: Conditional independence.}
By the Markov property (\cref{asp:markov}), for distinct children $j_1, j_2 \in \Children(i)$, the subtree valuations $\bmv_{\Subtree(j_1)} = \{v_m\}_{m \in \Subtree(j_1)}$ and $\bmv_{\Subtree(j_2)} = \{v_m\}_{m \in \Subtree(j_2)}$ satisfy:
\[
\calF(\bmv_{\Subtree(j_1)}, \bmv_{\Subtree(j_2)} \mid v_i) = \calF(\bmv_{\Subtree(j_1)} \mid v_i) \cdot \calF(\bmv_{\Subtree(j_2)} \mid v_i),
\]
since the tree path between any $m_1 \in \Subtree(j_1)$ and $m_2 \in \Subtree(j_2)$ passes through $i$.

\paragraph{Step 2: Factorization of conditional trade indicators.}
For any $m \in \Subtree(j)$ where $j \in \Children(i)$, we show $y_{m \mid i} = x_j \cdot y_{m \mid j}$. By \cref{eq:y-conditional}:
\begin{align*}
y_{m \mid i} &= \prod_{\substack{\ell \in \Path(m) \cup \{m\} \\ \ell \notin \Path(i) \cup \{i\}}} x_\ell.
\end{align*}
Since $j \in \Children(i)$, we have $\Path(j) \cup \{j\} = \Path(i) \cup \{i,j\}$, and:
\[
\big(\Path(m) \cup \{m\}\big) \setminus \big(\Path(i) \cup \{i\}\big) = \{j\} \cup \Big[\big(\Path(m) \cup \{m\}\big) \setminus \big(\Path(j) \cup \{j\}\big)\Big].
\]
Therefore $y_{m \mid i} = x_j \cdot y_{m \mid j}$. In the special case $m = j$: $y_{j \mid i} = x_j$ and $y_{j \mid j} = 1$.

\paragraph{Step 3: Gain decomposition.}
The total expected selling profit of player $i$ decomposes as follows. The descendants of $i$ partition as $\Desc(i) = \bigsqcup_{j \in \Children(i)} \Subtree(j)$. Using this partition:
\begin{align}
&\bbE\bigg[\sum_{m \in \Desc(i)} \pi(m,i) \cdot p_m \cdot y_{m \mid i} \;\bigg|\; v_i\bigg] \nonumber\\
&= \sum_{j \in \Children(i)} \bbE\bigg[\sum_{m \in \Subtree(j)} \pi(m,i) \cdot p_m \cdot x_j \cdot y_{m \mid j} \;\bigg|\; v_i\bigg] \label{eq:prf:decomp:partition}\\
&= \sum_{j \in \Children(i)} \bbE_{\bmv_{\Subtree(j)} \sim \calF(\cdot | v_i)}\bigg[\sum_{m \in \Subtree(j)} \pi(m,i) \cdot p_m \cdot x_j \cdot y_{m \mid j}\bigg], \label{eq:prf:decomp:indep}
\end{align}
where \eqref{eq:prf:decomp:partition} uses Step~2, and \eqref{eq:prf:decomp:indep} uses Step~1 (the summand for subtree $j$ depends only on $\bmv_{\Subtree(j)}$). By definition (\cref{eq:gain-child}), this equals $\sum_{j \in \Children(i)} G^{(j)}_i(v_i, p_j)$. Combined with the utility expansion~\eqref{eq:prf:simplify:factor} from the proof of \cref{lem:simplify-sell}, the conditional expected utility of player $i$ is
\begin{equation}
\label{eq:lem2:decomp}
\bbE_{\bmv_{-i} \sim \calF(\cdot \mid v_i)}\!\big[U_i(h(\tp_i, \tbmp^*_{-i}, \tx_i, \tbmx^*_{-i}; \bmv); v_i) \mid h_i\big] = x_i \cdot y_{\Parent(i)} \cdot \Big[(v_i - p_i) + \sum_{j \in \Children(i)} G^{(j)}_i(v_i, p_j)\Big],
\end{equation}
where $x_i = \tx_i(h_i, v_i) \in \{0,1\}$ and $p_j = \tp_i(h_i, j, v_i)$ are player $i$'s (possibly deviating) actions.

\paragraph{Step 4: Optimal actions and sequential rationality.}
We show the right-hand side of~\eqref{eq:lem2:decomp} is maximized by the equilibrium actions prescribed by $(\tbmp^*, \tbms^*)$, for every realization of $h_i$. As $x_i \cdot y_{\Parent(i)} \ge 0$, each term $G^{(j)}_i(v_i, p_j)$ depends only on $p_j$ (given $v_i$), and the children's prices enter the bracket additively and independently, the bracket is maximized by choosing $p_j \in \argmax_{p_j \ge 0} G^{(j)}_i(v_i, p_j)$ for every $j$, attaining the value
\[
\Phi_i(v_i) \;\coloneqq\; (v_i - p_i) + \sum_{j \in \Children(i)} \max_{p_j \ge 0} G^{(j)}_i(v_i, p_j);
\]
this is exactly the pricing rule~\eqref{eq:pricing-decomp}. Given these prices, the prefactor $x_i \cdot y_{\Parent(i)}$ is maximized over $x_i \in \{0,1\}$ by $x_i = \one\{\Phi_i(v_i) \ge 0\}$ when $y_{\Parent(i)} = 1$, and the value is $0$ regardless of $x_i$ when $y_{\Parent(i)} = 0$. By the threshold rule~\eqref{eq:threshold-decomp}, $\ts^*_i(v_i) = v_i + \sum_{j} \max_{p_j \ge 0} G^{(j)}_i(v_i, p_j)$, so $\Phi_i(v_i) \ge 0 \iff \ts^*_i(v_i) \ge p_i$; hence the optimal buying action coincides with the equilibrium threshold action $x_i = \one\{\ts^*_i(v_i) \ge p_i\}$.

Therefore, at \emph{every} information set $(h_i, v_i)$, the prescribed equilibrium actions maximize the conditional expected utility in~\eqref{eq:lem2:decomp}, which is precisely $\bbU_i\big((\tbmp^*, \tbms^*), \calF(\cdot \mid v_i) \mid h_i, v_i\big)$ of~\eqref{eq:cond-utility}; that is, for every deviation $(\tp_i, \tx_i)$,
\[
\bbU_i\big((\tbmp^*, \tbms^*), \calF(\cdot \mid v_i) \mid h_i, v_i\big) \;\ge\; \bbU_i\big((\tp_i, \tbmp^*_{-i}, \tx_i, \tbmx^*_{-i}), \calF(\cdot \mid v_i) \mid h_i, v_i\big).
\]
This is the sequential-rationality requirement~(a) of \cref{def:pbe}. As $i$, $v_i$, and $h_i$ were arbitrary, and the beliefs $\calF(\cdot \mid v_i)$ are Bayes-consistent (\cref{fact:simplify}), the profile $(\tbmp^*, \tbms^*)$ is a perfect Bayesian equilibrium.
\qed
\end{proof}

\subsection{Proof of \texorpdfstring{\cref{thm:algo-exact}}{Theorem~\ref{thm:algo-exact}}: Polynomial Algorithm}
\thmAlgoExact*
\begin{proof}
We prove that \cref{alg:exact,alg:forward-dp} correctly compute a perfect Bayesian equilibrium in $O(N d\, S^4)$ time, where $d$ is the depth of $\calT$.

\paragraph{Price Simplification.}
For a fixed non-leaf player $i$ and child $j \in \Children(i)$, define the finite set of candidate prices:
\begin{equation}
\label{eq:prf:algo:candidates}
\calP_j \;=\; \big\{s^*(j, \ell) : \ell \in [S]\big\}.
\end{equation}
We show that the optimal price $p^*_j$ must lie in $\calP_j$ (or be $0$ if all values in $\calP_j$ are negative).

\emph{Case 1:} All elements in $\calP_j$ are negative. For any $p_j \ge 0$ and any $\ell \in [S]$: $\ts^*_j(w_\ell) = s^*(j,\ell) < 0 \le p_j$, so $x_j = \one\{s^*(j,\ell) \ge p_j\} = 0$. Therefore $G^{(j)}_i(w_k, p_j) = 0$ for all $p_j \ge 0$, and setting $p_j = 0$ is optimal.

\emph{Case 2:} $\calP_j$ contains nonneg elements. Let $s_{\max} = \max \calP_j$. Setting $p_j > s_{\max}$ gives $G^{(j)}_i(w_k, p_j) = 0$, weakly dominated by $p_j = s_{\max}$. For $p_j \le s_{\max}$, consider the effect of increasing $p_j$ within an interval $(s^*(j, \ell_a), s^*(j, \ell_b))$ between two consecutive values in $\calP_j$. In this interval, the set of buying types $\{\ell : s^*(j,\ell) \ge p_j\}$ remains constant, while the per-trade revenue $\pi(j,i) \cdot p_j$ weakly increases. Hence:
\begin{equation}
\label{eq:prf:algo:monotone}
G^{(j)}_i(w_k, s^*(j,\ell_b)) \;\ge\; G^{(j)}_i(w_k, p_j) \quad \text{for all } p_j \in (s^*(j,\ell_a), s^*(j,\ell_b)).
\end{equation}
This shows that the optimal $p^*_j$ lies in $\calP_j$.

\paragraph{Backward Dynamic Program.}
\cref{alg:exact} processes the tree bottom-up with terminal condition $s^*(i,k) = w_k$ for all leaves $i \in \calL$ and $k \in [S]$. For each non-leaf node $i$ (processed after all descendants), each valuation $w_k$, and each child $j \in \Children(i)$, the algorithm enumerates candidate prices $p_j = s^*(j, k')$ for $k' \in \calP_j$ and computes the gain $G^{(j)}_i(w_k, s^*(j,k'))$ via the Forward DP (\cref{alg:forward-dp}). The optimal price is then:
\begin{equation}
\label{eq:prf:algo:optimal}
l^*(j,k) \in \argmax_{k' \in \calP_j} G^{(j)}_i(w_k, s^*(j,k')), \qquad p^*(j,k) = s^*(j, l^*(j,k)).
\end{equation}
After processing all children, the threshold value is updated:
\begin{equation}
\label{eq:prf:algo:threshold}
s^*(i,k) = w_k + \sum_{j \in \Children(i)} G^{(j)}_i\big(w_k, p^*(j,k)\big).
\end{equation}

\paragraph{Computing the Gain Function via Forward DP.}
The core difficulty is computing $G^{(j)}_i(w_k, p_j)$, which involves the expected reallocated profit from all trades in $\Subtree(j)$. By \cref{eq:gain-child}:
\begin{equation}
\label{eq:prf:algo:gain-expand}
G^{(j)}_i(w_k, p_j) = \bbE_{\bmv_{\Subtree(j)} \sim \calF(\cdot|w_k)}\bigg[\sum_{m \in \Subtree(j)} \pi(m,i) \cdot p_m \cdot x_j \cdot y_{m|j}\bigg].
\end{equation}
Naively evaluating this sum over all possible valuation profiles $\bmv_{\Subtree(j)}$ leads to exponential complexity. We overcome this using a Forward DP.

Define the event $Y(m, \ell)$ for $m \in \Subtree(j)$ and $\ell \in [S]$ as:
\[
Y(m, \ell) \;\coloneqq\; \{v_m = w_\ell\} \cap \{y_{m|j} \cdot x_j = 1\},
\]
\ie, player $m$ has valuation $w_\ell$ and all trades from $j$ down to $m$ (including both $j$ and $m$) have occurred.

\emph{Initialization (trade $j$).} For $\ell \in [S]$, given $v_i = w_k$ and price $p_j = s^*(j,k')$:
\begin{equation}
\label{eq:prf:algo:Y-init}
\Pr[Y(j, \ell)] = \Pr[v_j = w_\ell, x_j = 1 \mid v_i = w_k] = q_{j,k,\ell} \cdot \one\{s^*(j,\ell) \ge s^*(j,k')\}.
\end{equation}
The second equality uses: $v_j = w_\ell$ has probability $q_{j,k,\ell}$ conditioned on $v_i = w_k$ (by \cref{asp:markov}), and $x_j = \one\{s^*(j,\ell) \ge p_j\} = \one\{s^*(j,\ell) \ge s^*(j,k')\}$ by the threshold strategy (\cref{cor:simplify-buy}).

\emph{Recursion (from $m$ to its child $m'$).} For each $m \in \Subtree(j)$, $m' \in \Children(m)$, and $\ell' \in [S]$:
\begin{equation}
\label{eq:prf:algo:Y-recurse}
\begin{aligned}
\Pr[Y(m', \ell')] &= \Pr[v_{m'} = w_{\ell'},\; y_{m'|j} \cdot x_j = 1] \\
&= \sum_{\ell \in [S]} \Pr[v_{m'} = w_{\ell'},\; x_{m'} = 1,\; Y(m,\ell)] \\
&= \sum_{\ell \in [S]} \Pr[v_{m'} = w_{\ell'},\; x_{m'} = 1 \mid Y(m,\ell)] \cdot \Pr[Y(m,\ell)] \\
&= \sum_{\ell \in [S]} \Pr[v_{m'} = w_{\ell'},\; x_{m'} = 1 \mid v_m = w_\ell] \cdot \Pr[Y(m,\ell)] \\
&= \sum_{\ell \in [S]} q_{m',\ell,\ell'} \cdot \one\{s^*(m',\ell') \ge p^*(m',\ell)\} \cdot \Pr[Y(m,\ell)].
\end{aligned}
\end{equation}
The fourth equality uses the Markov property: conditioned on $v_m = w_\ell$, the event $\{v_{m'} = w_{\ell'}, x_{m'} = 1\}$ is independent of the events involving ancestors of $m$. The fifth equality uses $\Pr[v_{m'} = w_{\ell'} \mid v_m = w_\ell] = q_{m',\ell,\ell'}$ and $x_{m'} = \one\{s^*(m',\ell') \ge p^*(m',\ell)\}$.

\paragraph{Accumulating the Gain.}
We express $G^{(j)}_i(w_k, p_j)$ using the $Y$-probabilities. For each edge $(m, m')$ with $m \in \Subtree(j)$ and $m' \in \Children(m)$, define the one-shot revenue when $v_m = w_\ell$:
\begin{equation}
\label{eq:prf:algo:oneshot}
r(m', \ell) \;\coloneqq\; p^*(m', \ell) \cdot \sum_{\ell' \in [S]} q_{m',\ell,\ell'} \cdot \one\{s^*(m', \ell') \ge p^*(m', \ell)\}.
\end{equation}
The contribution of trade $m'$ to the gain is:
\begin{align}
&\bbE\big[\pi(m',i) \cdot p^*(m', v_m) \cdot x_j \cdot y_{m'|j} \;\big|\; v_i = w_k\big] \nonumber\\
&= \sum_{\ell \in [S]} \pi(m',i) \cdot p^*(m', \ell) \cdot \Pr[v_m = w_\ell,\; x_{m'} = 1,\; y_{m|j} \cdot x_j = 1] \nonumber\\
&= \sum_{\ell \in [S]} \pi(m',i) \cdot p^*(m', \ell) \cdot \Pr[Y(m,\ell)] \cdot \sum_{\ell'} q_{m',\ell,\ell'} \cdot \one\{s^*(m',\ell') \ge p^*(m',\ell)\} \nonumber\\
&= \sum_{\ell \in [S]} \pi(m',i) \cdot r(m', \ell) \cdot \Pr[Y(m,\ell)]. \label{eq:prf:algo:contrib}
\end{align}
For the root of the subtree ($m' = j$, with $m = i$), we use $p_j = s^*(j,k')$:
\begin{equation}
\label{eq:prf:algo:root-contrib}
\text{root contribution} = \pi(j,i) \cdot s^*(j,k') \cdot \sum_{\ell \in [S]} \Pr[Y(j,\ell)].
\end{equation}
Summing~\eqref{eq:prf:algo:contrib} over all edges in $\Subtree(j)$ and adding~\eqref{eq:prf:algo:root-contrib}:
\begin{equation}
\label{eq:prf:algo:gain-formula}
G^{(j)}_i(w_k, s^*(j,k')) = \pi(j,i) \cdot s^*(j,k') \cdot \sum_{\ell} \Pr[Y(j,\ell)] + \sum_{\substack{m \in \Subtree(j) \\ m' \in \Children(m)}} \sum_{\ell \in [S]} \pi(m',i) \cdot r(m',\ell) \cdot \Pr[Y(m,\ell)].
\end{equation}
This is precisely the computation in \cref{alg:forward-dp}.

\paragraph{Correctness.}
By construction, the strategies $\ts^*_i(w_k) = s^*(i,k)$ and $\tp^*_j(w_k) = p^*(j,k)$ satisfy $p^*(j,k) \in \argmax_{p_j \in \calP_j} G^{(j)}_i(w_k, p_j)$ for each edge $(i,j)$ and each $k \in [S]$. By \eqref{SE:transform}, this strategy profile constitutes a perfect Bayesian equilibrium: each seller's pricing maximizes $G^{(j)}_i$ independently across children (\cref{lem:subtree-decomp}), and the threshold strategy is a best response (\cref{cor:simplify-buy}), so the prescribed actions are optimal at every information set.

\paragraph{Complexity.}
The recursion~\eqref{eq:prf:algo:Y-recurse} processes each edge $(m, m')$ in $\Subtree(j)$ in $O(S^2)$ time (summing over $\ell$ for each $\ell'$). Over all $|\Subtree(j)|$ nodes, the Forward DP takes $O(|\Subtree(j)| \cdot S^2)$ time. This is repeated for $S$ seller valuations and at most $S$ candidate prices, giving $O(|\Subtree(j)| \cdot S^4)$ per edge $(i,j)$. The Backward DP sums over all edges:
\[
\sum_{j \in \calI \setminus \{r\}} O(|\Subtree(j)| \cdot S^4) \;=\; O\bigg(\sum_{j \in \calI \setminus \{r\}} |\Subtree(j)|\bigg) \cdot S^4 \;=\; O(N d\, S^4),
\]
using $\sum_{j \neq r} |\Subtree(j)| = \sum_{m \in \calI} |\Path(m)| = \sum_{m} \mathrm{depth}(m) \le N d$: each node $m$ is counted once for each $j \neq r$ with $m \in \Subtree(j)$---namely $m$ itself and each of its strict ancestors below the root, $\mathrm{depth}(m)$ in total---where $d = \max_m \mathrm{depth}(m)$ is the depth of $\calT$, at most $N$.
\qed
\end{proof}

\subsection{Proof of \texorpdfstring{\cref{thm:algo-approx}}{Theorem~\ref{thm:algo-approx}}: FPTAS}
\thmAlgoApprox*
\begin{proof}
The algorithm has three steps: (i)~discretize the continuous valuation supports, (ii)~compute an exact equilibrium of the discretized model via \cref{thm:algo-exact}, and (iii)~lift the discrete strategies to continuous strategies.

\paragraph{Step 1: Discretization.}
Fix a discretization precision $\varepsilon_0 > 0$ (to be determined). Define the grid:
\[
S \;:=\; \big\lfloor \varepsilon_0^{-1}(H_0 - L_0)\big\rfloor + 1, \qquad w_\ell \;:=\; \ell \cdot \varepsilon_0 + L_0,\quad 0 \le \ell \le S.
\]
These $S+1$ points $w_0, \ldots, w_S$ define $S$ discrete valuations (the states of the discretized model are indexed by $\ell \in \{0,\ldots,S-1\}$). Since $S \ge \varepsilon_0^{-1}(H_0 - L_0)$, the top point satisfies $w_S = S\varepsilon_0 + L_0 \ge H_0$, so the cells $I_\ell = [w_\ell, w_{\ell+1})$ for $0 \le \ell \le S-1$ cover $[L_0, H_0]$ and $F_i(w_S \mid \cdot) = 1$; in particular the boundary transition $q_{i,k,S-1}$ is well defined. Construct a discretized model by defining transition probabilities between grid points: for each non-root $i$ and $k, \ell \in \{0, \ldots, S-1\}$,
\[
q_{i,k,\ell} \;:=\; F_i(w_{\ell+1} \mid w_k) - F_i(w_\ell \mid w_k).
\]
Denote by $\hat{\calF}$ the distribution induced by $Q = \{q_{i,k,\ell}\}$.

\paragraph{Step 2: Exact Equilibrium of Discretized Model.}
Apply \cref{thm:algo-exact} to the discretized model $(\calT, \pi, Q)$ to obtain exact equilibrium values $\{s^*(i,k)\}$ and $\{p^*(i,k)\}$ for all $i \in \calI \setminus \{r\}$ and $k \in \{0,\ldots,S-1\}$.

\paragraph{Step 3: Lifting to Continuous Strategies.}
For any continuous valuation $v_i \in [L_0, H_0]$, define its grid index $k(v_i) := \lfloor \varepsilon_0^{-1}(v_i - L_0)\rfloor$, so that $w_{k(v_i)} \le v_i < w_{k(v_i)+1}$. The lifted strategies are piecewise constant:
\[
\ts^*_i(v_i) \;:=\; s^*\big(i, k(v_i)\big), \qquad \tp^*_j(v_i) \;:=\; p^*\big(j, k(v_i)\big).
\]
Since the lifted threshold $\ts^*_j$ is piecewise constant over grid cells, the optimal deviation for player $i$ can be restricted to the finite set $\{s^*(j,\ell) : 0 \le \ell < S\}$ without loss of generality, by the same price candidate argument as in \cref{thm:algo-exact}.

\paragraph{Error Bound.}
The core step is bounding the approximation error introduced by discretization. Write $G^{(j)}_i(v_i, p_j)$ for the gain under the true distribution $\calF$, and $G^{(j)}_{\hat{\calF},i}(w_k, p_j)$ for the gain under the discretized distribution $\hat{\calF}$. The gain function expands as:
\[
G^{(j)}_i(v_i, p_j) = \bbE_{\bmv_{\Subtree(j)} \sim \calF(\cdot|v_i)}\!\big[R^{(j)}_i(p_j, \bmv_{\Subtree(j)})\big],
\]
where $R^{(j)}_i$ is the profit from trades in $\Subtree(j)$ given the lifted strategies. Since the lifted strategies are piecewise constant over grid cells, $R^{(j)}_i(\bmv_{\Subtree(j)})$ depends on $\bmv_{\Subtree(j)}$ only through the cell indices $\{k(v_m)\}_{m \in \Subtree(j)}$.

\emph{Probability ratio bound.} Using the log-Lipschitz condition (\cref{asp:lipschitz}), for any node $m$ and conditioning value $v_{\Parent(m)} \in [w_k, w_{k+1})$:
\[
\big|\log f_m(v_m \mid v_{\Parent(m)}) - \log f_m(v_m \mid w_k)\big| \;\le\; K_0 \cdot |v_{\Parent(m)} - w_k| \;\le\; K_0 \varepsilon_0.
\]
For a path $m_1, m_2, \ldots, m_d$ in $\Subtree(j)$ with $v_{m_l} \in [w_{k_l}, w_{k_l+1})$, the joint density ratio satisfies:
\begin{equation}
\label{eq:prf:fptas:density-ratio}
e^{-d K_0 \varepsilon_0} \;\le\; \frac{\prod_{l=1}^d f_{m_l}(v_{m_l} \mid v_{\Parent(m_l)})}{\prod_{l=1}^d f_{m_l}(v_{m_l} \mid w_{k(\Parent(m_l))})} \;\le\; e^{d K_0 \varepsilon_0}.
\end{equation}
Since any path in the tree has length at most $N$, integrating over cells and using the Markov factorization $\calF(\bmv_{\Subtree(j)} \mid v_i) = \prod_{m \in \Subtree(j)} f_m(v_m \mid v_{\Parent(m)})$:
\begin{equation}
\label{eq:fptas:prob-bound}
e^{-dK_0\varepsilon_0}\,\Pr_{\hat{\calF}}[A] \;\le\; \Pr_{\calF}[A \mid v_i] \;\le\; e^{dK_0\varepsilon_0}\,\Pr_{\hat{\calF}}[A \mid w_{k(v_i)}],
\end{equation}
for any event $A$ that depends on $\bmv_{\Subtree(j)}$ only through cell indices.

\emph{Gain bound.} Fix a node $i$ with $v_i \in [w_{k_i}, w_{k_i+1})$. Because the lifted strategies are piecewise constant over grid cells, for each child $j$ the optimum of $G^{(j)}_i(v_i, \cdot)$ over $\bbR_+$ is attained on the finite grid $\calP_j = \{s^*(j,\ell)\}_\ell$, for both the true and the discretized models (the candidate-price argument of \cref{thm:algo-exact}). As the per-subtree profit $R^{(j)}_i \ge 0$ depends on $\bmv_{\Subtree(j)}$ only through cell indices, \eqref{eq:fptas:prob-bound} gives, for every $p \in \calP_j$,
\begin{equation}
\label{eq:prf:fptas:gain-upper}
e^{-d K_0 \varepsilon_0}\, G^{(j)}_{\hat{\calF},i}(w_{k_i}, p) \;\le\; G^{(j)}_i(v_i, p) \;\le\; e^{d K_0 \varepsilon_0}\, G^{(j)}_{\hat{\calF},i}(w_{k_i}, p).
\end{equation}
Writing $M_i \coloneqq \sum_{j} \max_{p \in \calP_j} G^{(j)}_i(v_i, p)$ and $\widehat{M}_i \coloneqq \sum_{j} \max_{p \in \calP_j} G^{(j)}_{\hat{\calF},i}(w_{k_i}, p)$ (both nonnegative), taking the max over $\calP_j$ and summing over $j$ yields
\begin{equation}
\label{eq:prf:fptas:M}
e^{-d K_0 \varepsilon_0}\, \widehat{M}_i \;\le\; M_i \;\le\; e^{d K_0 \varepsilon_0}\, \widehat{M}_i .
\end{equation}
By \cref{def:eps-equilibrium}, the best-response threshold under the lifted profile is $\ts^\dagger_i(v_i) = v_i + M_i$, while the algorithm outputs $\ts^*_i(v_i) = s^*(i,k_i) = w_{k_i} + \widehat{M}_i$ (the exact discretized equilibrium, \cref{eq:prf:algo:threshold}).

\emph{Threshold accuracy.} Let $R \coloneqq \max(|L_0|, |H_0|)$, so that $|v_i| \le R$. From~\eqref{eq:prf:fptas:M}, $|\widehat{M}_i - M_i| \le (e^{d K_0 \varepsilon_0} - 1) M_i$; combined with $0 \le v_i - w_{k_i} < \varepsilon_0$,
\[
|\ts^*_i(v_i) - \ts^\dagger_i(v_i)| = \big|(w_{k_i} - v_i) + (\widehat{M}_i - M_i)\big| \;\le\; \varepsilon_0 + (e^{d K_0 \varepsilon_0} - 1)\, M_i .
\]
Since $M_i = \ts^\dagger_i(v_i) - v_i \le \ts^\dagger_i(v_i) + R \le (1+R)\max\{1, \ts^\dagger_i(v_i)\}$ and $e^x - 1 \le 2x$ for $x \in [0,1]$,
\[
|\ts^*_i(v_i) - \ts^\dagger_i(v_i)| \;\le\; \varepsilon_0\big(1 + 2 d K_0 (1+R)\big)\max\{1, \ts^\dagger_i(v_i)\}.
\]

\emph{Pricing near-optimality.} The chosen price $\tp^*_j(v_i) = p^*(j,k_i)$ maximizes $G^{(j)}_{\hat{\calF},i}(w_{k_i}, \cdot)$ over $\calP_j$, so applying the two sides of~\eqref{eq:prf:fptas:gain-upper} gives $G^{(j)}_i(v_i, \tp^*_j(v_i)) \ge e^{-2 d K_0 \varepsilon_0} \max_{p \in \calP_j} G^{(j)}_i(v_i, p)$. With $1 - e^{-y} \le y$,
\[
\max_{p \ge 0} G^{(j)}_i(v_i, p) - G^{(j)}_i(v_i, \tp^*_j(v_i)) \;\le\; 2 d K_0 \varepsilon_0\, M_i \;\le\; 2 d K_0 \varepsilon_0 (1+R)\max\{1, \ts^\dagger_i(v_i)\}.
\]

\emph{Choice of $\varepsilon_0$.} Setting
\begin{equation}
\label{eq:prf:fptas:eps0}
\varepsilon_0 = \frac{\min(\varepsilon,1)}{8\,(1 + d K_0)\,(1 + R)}, \qquad R = \max(|L_0|, |H_0|),
\end{equation}
ensures $d K_0 \varepsilon_0 \le \tfrac18 \le 1$ and bounds both bracketed factors above by $\varepsilon$. Hence conditions~(a) and~(b) of \cref{def:eps-equilibrium} hold for every player, so the lifted profile is an $\varepsilon$-approximate equilibrium. (The definition $R = \max(|L_0|,|H_0|)$ covers $H_0 < 0$ as well, since then all valuations lie in $[L_0, H_0]$ with $|v_i| \le |L_0| = R$.)

\paragraph{Complexity.}
The grid size is $S = O(\varepsilon_0^{-1}(H_0 - L_0)) = O(dK_0(H_0-L_0)(1+|L_0|)/\varepsilon)$. The exact algorithm runs in $O(N d\, S^4)$ time, giving overall complexity $\poly(N, \varepsilon^{-1}, H_0-L_0, K_0)$.
\qed
\end{proof}

\subsection{Homogeneity Scaling Lemma}
\begin{lemma}[Homogeneity Scaling]
\label{lem:scaling}
Let \cref{asp:markov,asp:posi,asp:homo} hold and let $(\tbmp^*, \tbms^*)$ be the simplified equilibrium under a PRM $\pi$ with the smallest-maximizer tie-break.
\emph{(1) Linearity:} for every $i \in \calI \setminus \{r\}$ there is a constant $\hat{s}_i \ge 1$ with $\ts^*_i(v) = \hat{s}_i\, v$ for all $v > 0$, and for every $j \in \calI \setminus \{r\}$ a constant $\hat{p}_j \ge 0$ with $\tp^*_j(v) = \hat{p}_j\, v$ for all $v > 0$.
\emph{(2) Mechanism-independent trade events:} if the seller $i = \Parent(j)$ collects from $\Subtree(j)$ only the one-shot payment of trade $j$---\ie, $\pi(m,i) = 0$ for all $m \in \Desc(j)$---then $E_j(\tbmp^*, \tbms^*) = \{\bmv : v_j / v_{\Parent(j)} \ge \tau_j\}$, where the threshold $\tau_j$ depends only on the conditional law $F_j$, not on $\pi$.
\end{lemma}
\begin{proof}
\smallskip\noindent\textit{Part 1 (Linearity).}\enspace
We argue by backward induction on the tree that every equilibrium threshold $\ts^*_i$ and every posted price $\tp^*_j$ is homogeneous of degree one. For a leaf $i$, $\ts^*_i(v) = v = \hat{s}_i v$ with $\hat{s}_i = 1$. Fix a non-leaf $i$ and assume the claim for every node in $\Desc(i)$. Fix a child $j \in \Children(i)$ and recall from \cref{eq:gain-child} that, writing $\bmw \coloneqq \bmv_{\Subtree(j)}$,
\[
G^{(j)}_i(v_i, p) = \bbE_{\bmw \sim \calF(\cdot \mid v_i)}\big[\Psi_p(\bmw)\big],
\qquad
\Psi_p(\bmw) \coloneqq \one\{\ts^*_j(w_j) \ge p\}\Big(\pi(j,i)\,p + \!\!\sum_{m \in \Desc(j)}\!\! \pi(m,i)\, \tp^*_m(w_{\Parent(m)})\, y_{m \mid j}(\bmw)\Big),
\]
where $y_{m\mid j}(\bmw) = \prod_{\ell}\one\{\ts^*_\ell(w_\ell) \ge \tp^*_\ell(w_{\Parent(\ell)})\}$ ranges over the edges on the path from $j$ to $m$ (\cref{eq:y-conditional}). We establish $G^{(j)}_i(\alpha v_i, \alpha p) = \alpha\, G^{(j)}_i(v_i, p)$ for every $\alpha > 0$ in two steps.

\emph{Step 1 (the conditional law scales).} We claim that if $\bmw \sim \calF(\cdot \mid v_i)$ then $\alpha \bmw \sim \calF(\cdot \mid \alpha v_i)$. By the Markov property (\cref{asp:markov}), the conditional law of $\bmw$ given $v_i$ factorizes along the subtree,
\[
\calF(\bmw \mid v_i) = f_j(w_j \mid v_i)\!\!\prod_{m \in \Desc(j)}\!\! f_m\big(w_m \mid w_{\Parent(m)}\big),
\]
with the root factor conditioned on $v_i$. By homogeneity (\cref{asp:homo}), each conditional CDF satisfies $F_m(\alpha x \mid \alpha y) = F_m(x \mid y)$; equivalently, if $W \sim F_m(\cdot \mid y)$ then $\alpha W \sim F_m(\cdot \mid \alpha y)$. Applying this to the top factor ($W = w_j$, $y = v_i$) scales $w_j \mapsto \alpha w_j$ and its conditioning value $v_i \mapsto \alpha v_i$; the scaled $\alpha w_j$ is in turn the conditioning value for the factors of its children, so the rescaling propagates down the subtree by induction on depth. Hence the joint law of $\alpha\bmw$ is exactly $\calF(\cdot \mid \alpha v_i)$, proving the claim.

\emph{Step 2 (the integrand scales).} We claim $\Psi_{\alpha p}(\alpha \bmw) = \alpha\,\Psi_p(\bmw)$ pointwise. Every indicator is invariant: by the inductive hypothesis $\ts^*_\ell(\alpha w_\ell) = \alpha\ts^*_\ell(w_\ell)$ and $\tp^*_\ell(\alpha w_{\Parent(\ell)}) = \alpha\tp^*_\ell(w_{\Parent(\ell)})$ for every $\ell \in \Desc(i)$, so
\[
\one\{\ts^*_j(\alpha w_j) \ge \alpha p\} = \one\{\ts^*_j(w_j) \ge p\},
\qquad
\one\{\ts^*_\ell(\alpha w_\ell) \ge \tp^*_\ell(\alpha w_{\Parent(\ell)})\} = \one\{\ts^*_\ell(w_\ell) \ge \tp^*_\ell(w_{\Parent(\ell)})\},
\]
whence $\one\{\ts^*_j(\alpha w_j)\ge\alpha p\} = \one\{\ts^*_j(w_j)\ge p\}$ and $y_{m\mid j}(\alpha\bmw) = y_{m\mid j}(\bmw)$. The monetary terms scale by $\alpha$: $\pi(j,i)(\alpha p) = \alpha\,\pi(j,i)\,p$ and $\pi(m,i)\tp^*_m(\alpha w_{\Parent(m)}) = \alpha\,\pi(m,i)\tp^*_m(w_{\Parent(m)})$. As $\Psi_{\alpha p}(\alpha\bmw)$ is the (unchanged) indicator times the ($\alpha$-scaled) monetary bracket, the claim follows.

Combining the two steps via the change of variables $\bmu = \alpha\bmw$,
\[
G^{(j)}_i(\alpha v_i, \alpha p)
= \bbE_{\bmu \sim \calF(\cdot \mid \alpha v_i)}\big[\Psi_{\alpha p}(\bmu)\big]
\overset{\text{Step 1}}{=} \bbE_{\bmw \sim \calF(\cdot \mid v_i)}\big[\Psi_{\alpha p}(\alpha\bmw)\big]
\overset{\text{Step 2}}{=} \alpha\,\bbE_{\bmw \sim \calF(\cdot \mid v_i)}\big[\Psi_p(\bmw)\big]
= \alpha\, G^{(j)}_i(v_i, p).
\]
Consequently $\max_{p \ge 0} G^{(j)}_i(\alpha v_i, p) = \alpha \max_{p \ge 0} G^{(j)}_i(v_i, p)$ (substitute $p = \alpha p'$), and the argmax set scales by $\alpha$, so its smallest element does too; hence $\tp^*_j(\alpha v_i) = \alpha\, \tp^*_j(v_i)$. Setting $\alpha = 1/v_i$ gives $\tp^*_j(v) = \hat{p}_j v$ with $\hat{p}_j = \tp^*_j(1) \ge 0$. Likewise, by the threshold decomposition~\eqref{eq:threshold-decomp}, $\ts^*_i(v_i) = v_i + \sum_{j \in \Children(i)} \max_{p \ge 0} G^{(j)}_i(v_i, p)$ is a sum of degree-one homogeneous functions, so $\ts^*_i(v) = \hat{s}_i v$ with $\hat{s}_i = \ts^*_i(1) \ge 1$. This closes the induction.

\smallskip\noindent\textit{Part 2 (Mechanism-independent trade events).}\enspace
Suppose $\pi(m,i) = 0$ for all $m \in \Desc(j)$. Then the downstream sum in $G^{(j)}_i$ vanishes and, using $\ts^*_j(v_j) = \hat{s}_j v_j$ from Part~1,
\[
G^{(j)}_i(v_i, p) = \pi(j,i)\, p \, \Pr[\hat{s}_j v_j \ge p \mid v_i] = \pi(j,i)\, p\, \big(1 - F_j(p/\hat{s}_j \mid v_i)\big).
\]
By \cref{asp:homo}, $F_j(p/\hat{s}_j \mid v_i) = F_j\big(p/(\hat{s}_j v_i) \mid 1\big)$. Substituting $p = \hat{s}_j v_i\, t$,
\[
G^{(j)}_i(v_i, \hat{s}_j v_i t) = \pi(j,i)\, \hat{s}_j v_i \cdot t\,\big(1 - F_j(t \mid 1)\big),
\]
whose maximizer over $t \ge 0$ is $\tau_j \in \argmax_{t \ge 0} t\,(1 - F_j(t \mid 1))$, depending only on $F_j$ (the positive factors $\pi(j,i)$ and $\hat{s}_j v_i$ do not affect the maximizer; the smallest-maximizer rule selects the same $\tau_j$ for every $v_i$). Hence $\tp^*_j(v_i) = \hat{s}_j \tau_j\, v_i$, and
\[
E_j(\tbmp^*, \tbms^*) = \{\ts^*_j(v_j) \ge \tp^*_j(v_i)\} = \{\hat{s}_j v_j \ge \hat{s}_j \tau_j v_i\} = \{v_j / v_i \ge \tau_j\},
\]
with $\tau_j$ depending only on $F_j$.
\qed
\end{proof}

\subsection{Proof of \texorpdfstring{\cref{thm:partial}}{Theorem~\ref{thm:partial}}: Local Comparison}
\thmPartial*
\begin{proof}
For $k \in \{1, 2\}$, abbreviate the trade-event of trade $m$ under mechanism $\pi^k$ as $E^k_m \coloneqq E_m(\tbmp^{k,*}, \tbms^{k,*})$. We prove Part~(1) ($E^1_m = E^2_m$ for $m \neq j$) and Part~(2) ($E^2_j \subseteq E^1_j$) separately.

\paragraph{Proof of Part~(1).}
We partition the trades $m \in \calI \setminus \{r\}$ with $m \neq j$ into four groups: (A) $m \in \Desc(j)$; (B) $m \in \Subtree(c)$ for a sibling $c \in \Children(i) \setminus \{j\}$; (C) $m \in \Path(i) \cup \{i\}$, the trades on the root-to-$i$ path together with trade $i$; and (D) the remaining trades, $m \notin \Subtree(i) \cup \Path(i)$. Since $\Desc(i) = \bigsqcup_{c \in \Children(i)} \Subtree(c)$ and $\Subtree(i) = \{i\} \cup \Desc(i)$, these four groups partition $\{m \neq j\}$.

\emph{Case A: $m \in \Desc(j)$.}
By condition~(a), $\pi^1(n, k) = \pi^2(n, k)$ for all $k \neq i$. By the backward induction in \cref{lem:simplify-sell}, the equilibrium strategies internal to $\Subtree(j)$---namely the thresholds $\{\ts_n\}_{n \in \Subtree(j)}$ and the prices $\{\tp_n\}_{n \in \Desc(j)}$ for trades whose seller $\Parent(n)$ lies in $\Subtree(j)$---are computed bottom-up from the gain functions $\{G^{(n')}_n\}_{n \in \Subtree(j),\, n' \in \Children(n)}$, which depend on $\pi$ only through the reallocations $\{\pi(n, k) : n \in \Desc(j),\, k \in \Subtree(j) \cap \Path(n)\}$ to nodes inside $\Subtree(j)$. As none of these involve player $i$, they are identical under $\pi^1$ and $\pi^2$, so
\begin{equation}
\label{eq:prf:partial:subtree-same}
\ts^{1,*}_n = \ts^{2,*}_n \;\; \forall n \in \Subtree(j), \qquad \tp^{1,*}_n = \tp^{2,*}_n \;\; \forall n \in \Desc(j).
\end{equation}
(The price $\tp_j$ offered by $i$ to $j$ is set by player $i$ and is \emph{not} internal to $\Subtree(j)$; it is treated in Part~(2).) For $m \in \Desc(j)$, both $\ts^*_m$ and $\tp^*_m$ appear in~\eqref{eq:prf:partial:subtree-same}, hence $E^1_m = \{\bmv : \ts^*_m(v_m) \ge \tp^*_m(v_{\Parent(m)})\} = E^2_m$.

\emph{Case B: $m \in \Subtree(c)$ for some $c \in \Children(i) \setminus \{j\}$.}
By \cref{lem:subtree-decomp}, both the equilibrium strategies within $\Subtree(c)$ and the price $\tp_c$ that $i$ offers to $c$ (the maximizer of $G^{(c)}_i$) depend on $\pi$ only through $\{\pi(n, k) : n \in \Subtree(c),\, k \in \calI\}$. Sibling subtrees are disjoint, so every $n \in \Subtree(c)$ satisfies $n \notin \Subtree(j)$; condition~(e) then gives $\pi^1(n, i) = \pi^2(n, i)$, while condition~(a) gives $\pi^1(n, k) = \pi^2(n, k)$ for all $k \neq i$. Hence $\pi^1(n, k) = \pi^2(n, k)$ for all $n \in \Subtree(c)$ and all $k$, so $G^{(c),1}_i = G^{(c),2}_i$ and the equilibria within $\Subtree(c)$ coincide, giving $E^1_m = E^2_m$. (Condition~(e) is exactly what rules out a spurious change in $i$'s pricing to its other children---the subtlety that distinguishes the tree from the chain, where no siblings exist.)

\emph{Case C: $m \in \Path(i) \cup \{i\}$.}
The seller in trade $m$ is $\Parent(m)$, and $\Parent(m) \in \Path(i)$: for $m \in \Path(i)$ the parent is a higher ancestor of $i$, and for $m = i$ it is $\Parent(i)$; in all cases it lies on the root-to-$i$ path. By condition~(d), both mechanisms reallocate to $\Parent(m)$ exactly as the baseline $\pi^0$, which pays only the direct seller; therefore $\pi^k(m', \Parent(m)) = 0$ for every strict descendant $m' \in \Desc(m)$ and $k \in \{1,2\}$. \cref{lem:scaling}(2), applied to trade $m$ under each mechanism, then yields $E^k_m = \{\bmv : v_m / v_{\Parent(m)} \ge \tau_m\}$ with $\tau_m$ depending only on $F_m$. Since $\tau_m$ is mechanism-independent, $E^1_m = E^2_m$.

\emph{Case D: $m \notin \Subtree(i) \cup \Path(i)$.}
Here $i \notin \Subtree(m)$, and $\Subtree(m)$ is disjoint from $\Subtree(i) \supseteq \Subtree(j)$. By condition~(a) the two mechanisms differ only in reallocation to $i$, and by conditions~(b),(c),(e) these differences occur only on trades in $\Subtree(j)$; none lie in $\Subtree(m)$. Hence both the threshold $\ts^*_m$ (computed within $\Subtree(m)$) and the price $\tp^*_m$ offered by $\Parent(m)$ (the maximizer of $G^{(m)}_{\Parent(m)}$, which by \cref{lem:subtree-decomp} depends only on reallocations within $\Subtree(m)$) are identical under $\pi^1$ and $\pi^2$, so $E^1_m = E^2_m$.

\paragraph{Proof of Part~(2).}
Since $\ts^{1,*}_j = \ts^{2,*}_j$ (proven in Part~(1)), it suffices to show $\tp^{1,*}_j(v_i) \le \tp^{2,*}_j(v_i)$ for all $v_i$, because then:
\[
E^2_j = \{\bmv : \ts^*_j(v_j) \ge \tp^{2,*}_j(v_i)\} \subseteq \{\bmv : \ts^*_j(v_j) \ge \tp^{1,*}_j(v_i)\} = E^1_j.
\]
We decompose the gain function $G^{(j),k}_i(v_i, p_j)$ under mechanism $\pi^k$ into one-shot and future components:
\begin{equation}
\label{eq:prf:partial:decomp}
G^{(j),k}_i(v_i, p_j) \;=\; r^k_i(v_i, p_j) \;+\; d^k_i(v_i, p_j),
\end{equation}
where
\begin{align*}
r^k_i(v_i, p_j) &\;\coloneqq\; \bbE\!\big[\pi^k(j,i) \cdot p_j \cdot \one\{\ts^*_j(v_j) \ge p_j\} \;\big|\; v_i\big], \\
d^k_i(v_i, p_j) &\;\coloneqq\; \bbE\!\bigg[\one\{\ts^*_j(v_j) \ge p_j\} \cdot \sum_{m \in \Desc(j)} \pi^k(m,i) \cdot p_m \cdot y_{m \mid j} \;\bigg|\; v_i\bigg].
\end{align*}
Here $r^k_i$ is the one-shot profit from trade $j$, and $d^k_i$ is the future reallocated profit from downstream trades. Note that the equilibrium strategies $\ts^*_j, \tp^*_m$ and the conditional trade indicators $y_{m \mid j}$ within $\Subtree(j)$ are the same under both mechanisms (proven above), so we drop the superscript $k$ from these quantities.

Let $p^k = \tp^{k,*}_j(v_i)$ denote the optimal price under $\pi^k$.

\emph{Degenerate case $\pi^1(j,i) = 0$.} Then $r^1_i \equiv 0$, so $G^{(j),1}_i(v_i, \cdot) = d^1_i(v_i, \cdot)$, which is weakly decreasing in $p$ (Property~2 below). The smallest-maximizer tie-break therefore selects $\tp^{1,*}_j(v_i) = 0 \le \tp^{2,*}_j(v_i)$, so the desired inequality holds. \emph{We assume henceforth $\pi^1(j,i) > 0$}, whence $\pi^2(j,i) \ge \pi^1(j,i) > 0$ by condition~(b). We establish two key monotonicity properties.

\emph{Property 1.} Since $\ts^{1,*}_j = \ts^{2,*}_j$ and $\pi^1(j,i), \pi^2(j,i) > 0$, we have $r^1_i(v_i, p) = \frac{\pi^1(j,i)}{\pi^2(j,i)}\, r^2_i(v_i, p)$ with coefficient $\frac{\pi^1(j,i)}{\pi^2(j,i)} \in (0, 1]$ by condition~(b). Consequently, for any $p, p'$:
\begin{equation}
\label{eq:prf:partial:prop1}
r^1_i(v_i, p) \ge r^1_i(v_i, p') \quad\Longrightarrow\quad \Delta r_i(v_i, p) \le \Delta r_i(v_i, p'),
\end{equation}
where $\Delta r_i \coloneqq r^1_i - r^2_i$: since the coefficient is \emph{strictly} positive, $r^1_i(v_i,p) \ge r^1_i(v_i,p')$ is equivalent to $r^2_i(v_i,p) \ge r^2_i(v_i,p')$, and $\Delta r_i = (\frac{\pi^1(j,i)}{\pi^2(j,i)} - 1) r^2_i$ has nonpositive coefficient, so the implication follows.

\emph{Property 2.} For $p' > p$, we have $\one\{\ts^*_j(v_j) \ge p'\} \le \one\{\ts^*_j(v_j) \ge p\}$ pointwise. Since prices are nonnegative and $y_{m|j} \in \{0,1\}$, the integrand in $d^k_i$ is nonnegative and is multiplied by $\one\{\ts^*_j(v_j) \ge p\}$, so $d^k_i(v_i, p)$ is weakly decreasing in $p$. For $\Delta d_i(v_i, p) \coloneqq d^1_i(v_i, p) - d^2_i(v_i, p)$:
\begin{align*}
\Delta d_i(v_i, p) &= \bbE\bigg[\one\{\ts^*_j(v_j) \ge p\} \cdot \sum_{m \in \Desc(j)} \big(\pi^1(m,i) - \pi^2(m,i)\big) \cdot p_m \cdot y_{m|j} \;\bigg|\; v_i\bigg].
\end{align*}
By condition~(c), $\pi^1(m,i) - \pi^2(m,i) \ge 0$ for all $m \in \Desc(j)$, so the summand is nonnegative; hence $\Delta d_i(v_i, p)$ is also weakly decreasing in $p$.

Now suppose for contradiction that $p^1 > p^2$. By optimality of $p^1$ under $\pi^1$:
\[
G^{(j),1}_i(v_i, p^1) \ge G^{(j),1}_i(v_i, p^2).
\]
Since $d^1_i$ is weakly decreasing in $p$ (Property~2) and $p^1 > p^2$, we have $d^1_i(v_i, p^1) \le d^1_i(v_i, p^2)$. Combined with the optimality inequality:
\begin{equation}
\label{eq:prf:partial:rstep}
r^1_i(v_i, p^1) \ge r^1_i(v_i, p^2) \quad\Longrightarrow\quad \Delta r_i(v_i, p^1) \le \Delta r_i(v_i, p^2),
\end{equation}
where the implication follows from Property~1. Define the quantity:
\begin{align*}
S \;&\coloneqq\; \big[G^{(j),2}_i(v_i, p^1) - G^{(j),2}_i(v_i, p^2)\big] - \big[G^{(j),1}_i(v_i, p^1) - G^{(j),1}_i(v_i, p^2)\big] \\
&=\; \big[\Delta r_i(v_i, p^2) - \Delta r_i(v_i, p^1)\big] + \big[\Delta d_i(v_i, p^2) - \Delta d_i(v_i, p^1)\big] \;\ge\; 0,
\end{align*}
where the first bracket is $\ge 0$ by~\eqref{eq:prf:partial:rstep} and the second bracket is $\ge 0$ by Property~2 (since $\Delta d_i$ is weakly decreasing and $p^1 > p^2$). Since $G^{(j),1}_i(v_i, p^1) \ge G^{(j),1}_i(v_i, p^2)$ and $S \ge 0$:
\[
G^{(j),2}_i(v_i, p^1) \ge G^{(j),2}_i(v_i, p^2).
\]
By optimality of $p^2$ under $\pi^2$, equality must hold: $G^{(j),2}_i(v_i, p^1) = G^{(j),2}_i(v_i, p^2)$. Then $S \ge 0$ forces $G^{(j),1}_i(v_i, p^1) = G^{(j),1}_i(v_i, p^2)$, so $p^2$ is also a maximizer of $G^{(j),1}_i(v_i, \cdot)$. But $p^2 < p^1$, and the smallest-maximizer tie-break selects the smallest maximizer under $\pi^1$; hence $p^1 = \tp^{1,*}_j(v_i) \le p^2$, contradicting $p^1 > p^2$. Therefore $\tp^{1,*}_j(v_i) \le \tp^{2,*}_j(v_i)$.
\qed
\end{proof}

\subsection{Proof of \texorpdfstring{\cref{cor:welfare}}{Corollary~\ref{cor:welfare}}: Welfare Improvement}
\begin{proof}
Fix any terminal history $h$. Summing the utilities in \cref{eq:utility-raw} over all players,
\[
\SW(h;\bmv) = \sum_{i \in \calI} U_i(h;v_i) = \sum_{i \in \calI} y_i v_i \;-\; \sum_{i} y_i p_i \;+\; \sum_{i} \sum_{j \in \Desc(i)} y_j\, \pi(j,i)\, p_j.
\]
Swapping the order of the last double sum and using that $j \in \Desc(i) \iff i \in \Path(j)$ together with the Path Only property (\cref{def:prm}),
\[
\sum_{i} \sum_{j \in \Desc(i)} y_j\, \pi(j,i)\, p_j = \sum_{j} y_j p_j \sum_{i \in \Path(j)} \pi(j,i) = \sum_{j} y_j p_j,
\]
where the last equality holds by budget balance. This cancels $\sum_i y_i p_i$, so $\SW(h;\bmv) = \sum_{i \in \calI} y_i v_i$. This is intuitive: with budget balance, monetary transfers are internal and welfare equals the total realized valuation.

Now set $h = h(\tbmp^*, \tbms^*; \bmv)$. Under the threshold equilibrium, trade $i$ occurs ($y_i = 1$) exactly when every player on the path from the root to $i$ is willing to buy, \ie, $\bmv \in \bigcap_{k \in (\Path(i) \cup \{i\}) \setminus \{r\}} E_k(\tbmp^*, \tbms^*)$. By \cref{thm:global}, $E_k(\tbmp^{0,*}, \tbms^{0,*}) \subseteq E_k(\tbmp^*, \tbms^*)$ for every $k$, so $y_i$ in $h(\tbmp^*, \tbms^*; \bmv)$ is weakly larger than $y_i$ in $h(\tbmp^{0,*}, \tbms^{0,*}; \bmv)$. Since $v_i > 0$ with probability 1, summing over $i$ and taking the expectation over $\bmv \sim \calF$ yields the claim.
\qed
\end{proof}

\newpage
\section{Supplementary Experimental Results}
\label{sec:app-experiments}

This appendix collects experimental results deferred from \cref{sec:exp} for space.

\subsection{Irregular-Tree Experiment}
\label{app:irregular}
Our complete-tree experiments in \cref{sec:exp} exploit level-symmetry to compute $d+1$ representative strategies. To confirm that the conclusions are not an artifact of this regularity, we additionally solve the \emph{irregular} tree of \cref{fig:tree-notation}---with unequal branching and depth---exactly and without any symmetry reduction, using the general-tree solver of \cref{thm:algo-exact}.
\Cref{tab:exp-irregular} reports, for this irregular instance, the same two metrics and four valuation models as \cref{tab:exp}. The qualitative picture is unchanged: every reallocation mechanism weakly beats RAW for all four models, OPT helps in all cases, and the instances of our general family again dominate the fixed-ratio TBDS---so the benefit of widening the design space is not special to symmetric trees.

\begin{table}[h]
\centering
\caption{Equilibrium transaction volume (TV) and social welfare (SW) on the \emph{irregular} tree of \cref{fig:tree-notation} (7 players, unequal branching and depth), across the four valuation models. Profit reallocation again improves on the baseline RAW; the empirically optimized OPT helps in every case, including M3 where the fixed-ratio TBDS and uniform AVE leave the small market unchanged.}
\label{tab:exp-irregular}
\begin{tabular*}{\textwidth}{@{\extracolsep{\fill}} l cccc cccc @{}}
\toprule
& \multicolumn{4}{c}{Transaction Volume} & \multicolumn{4}{c}{Social Welfare} \\
\cmidrule(lr){2-5}\cmidrule(lr){6-9}
Mechanism & M1 & M2 & M3 & M4 & M1 & M2 & M3 & M4 \\
\midrule
RAW  & $1.25$ & $2.78$ & $0.50$ & $1.46$ & $5.60$ & $0.78$ & $3.00$ & $0.55$ \\
TBDS & $1.89$ & $2.78$ & $0.50$ & $1.46$ & $7.05$ & $0.78$ & $3.00$ & $0.55$ \\
AVE  & $1.89$ & $2.85$ & $0.50$ & $1.53$ & $7.05$ & $0.80$ & $3.00$ & $0.57$ \\
OPT  & $2.96$ & $3.78$ & $0.61$ & $1.74$ & $8.27$ & $0.91$ & $3.31$ & $0.60$ \\
\bottomrule
\end{tabular*}
\end{table}

\end{document}